\documentclass[reqno,12pt,a4paper]{amsart}
\allowdisplaybreaks[4]
\usepackage{amssymb}
\usepackage{amsmath}
\usepackage{color}
\usepackage{hyperref}

\let\paragraph\S
\def\S#1{S_{(#1)}}

\def\P#1{P^{(#1)}}
\def\Q#1{Q^{(#1)}}
\def\Z#1{Z^{(#1)}}

\def\wt#1{\left|{#1}\right|}
\def\Zi#1{Z^{(#1)}}
\def\Ui#1{U^{(#1)}}
\def\Vi#1{V^{(#1)}}
\def\Wi#1{W^{(#1)}}
\def\Si#1{\mathcal{S}^{(#1)}}
\def\Ri#1{\mathcal{R}^{(#1)}}
\def\omegai#1{\omega^{(#1)}}
\def\Rii#1{R^{(#1)}}
\def\rhoi#1{\rho^{(#1)}}
\def\Xii#1{X^{(#1)}}
\def\Yii#1{Y^{(#1)}}
\def\Dii#1{D^{(#1)}}
\def\omegaib#1{\bar{\omega}^{(#1)}}
\def\Xiib#1{\bar{X}^{(#1)}}
\def\Yiib#1{\bar{Y}^{(#1)}}
\def\Diib#1{\bar{D}^{(#1)}}
\def\Riib#1{\bar{R}^{(#1)}}

\def\diff{\,\mathrm{d}}

\let\phi=\varphi
\def\psii#1{\psi^{(#1)}}
\def\Psii#1{\Psi^{(#1)}}
\def\phii#1{\phi^{(#1)}}

\def\PBS#1{\mathfrak S_{#1}}

\def\rsum{\mathop{\sum\nolimits^{\rlap{$\scriptstyle(\bullet)$}}}}
\def\hzav#1{\left[#1\right]}
\def\zav#1{\left(#1\right)}
\def\szav#1{\left\{#1\right\}}
\def\Res{\mathop{\rm Res}\nolimits}

\def\Int{\mathbb{I}}

\newtheorem{theorem}{Theorem}

\newtheorem{proposition}{Proposition}
\newtheorem{lemma}{Lemma}

\theoremstyle{definition}

\theoremstyle{remark}
\newtheorem{remark}{Remark}
\newtheorem*{example}{Example}
\newtheorem{fact}{Fact}

\newcommand*{\eval}[1]{\left.#1\right|}
\newcommand*{\abs}[1]{\left|#1\right|}
\newcommand{\const}{\mathrm{const}}

\newcommand{\cprime}{\/{\mathsurround=0pt$'$}}

\DeclareMathOperator{\sym}{sym}
\DeclareMathOperator{\cosym}{cosym}
\DeclareMathOperator{\Cl}{Cl}

\DeclareMathOperator{\Ev}{\mathbf{E}}

\def\rhoi#1{\rho^{(#1)}}

\usepackage[all]{xy}
\SelectTips{cm}{}
\usepackage{mathrsfs}
\let\mathcal\mathscr
\usepackage[mathscr]{eucal}

\begin{document}

\title{On symmetries of the Gibbons--Tsarev equation}

\author{H.~Baran} \address{Mathematical Institute in Opava, Silesian
  University in Opava, Na Rybn\'{\i}\v{c}ku 626/1, 746\,01 Opava, Czech
  Republic} \email{Hynek.Baran@math.slu.cz} \author{P.~Blaschke}
\address{Mathematical Institute in Opava, Silesian University in Opava, Na
  Rybn\'{\i}\v{c}ku 626/1, 746\,01 Opava, Czech Republic}
\email{Petr.Blaschke@math.slu.cz} \author{I.S.~Krasil{\cprime}shchik}
\address{V.A.~Trapeznikov Institute of Control Sciences RAS, Profsoyuznaya 65,
  117342 Moscow, Russia \& Independent University of Moscow, 119002, Bolshoy
  Vlasyevskiy Pereulok 11, Moscow} \email{josephkra@gmail.com}
\author{M.~Marvan} \address{Mathematical Institute in Opava, Silesian
  University in Opava, Na Rybn\'{\i}\v{c}ku 626/1, 746\,01 Opava, Czech
  Republic} \email{Michal.Marvan@math.slu.cz}

\date{\today}

\begin{abstract}
  We study the Gibbons--Tsarev equation 
  $z_{yy} + z_x z_{xy} - z_y z_{xx} + 1 = 0$ and, 
  using the known Lax pair, we construct infinite series of
  conservation laws and the algebra of nonlocal symmetries in the 
  covering associated with these conservation laws.
  We prove that the algebra is isomorphic to the Witt algebra. 
  Finally, we show that the constructed symmetries are
  unique in the class of polynomial ones.
\end{abstract}



\maketitle

\medskip

\section*{Introduction}
\label{sec:introduction}

The Gibbons--Tsarev equation considered in this paper was introduced
in~\cite{G-T 1} to classify finite reductions of the infinite Benney system.
The Gibbons--Tsarev equation is undoubtedly integrable~\cite{G-T 2, Y-G, B-G}.
It is known to have infinitely many conservation laws and infinitely many
symmetries (\cite[\paragraph~3.3]{G-L-R} and~\cite{H-K-M-V}).  However, unlike
the majority of integrable equations with two independent variables, the
Gibbons--Tsarev equation has only few local
symmetries~\cite[\paragraph~1]{Odess-Sokolov 2}, thus escaping symmetry-based
integrability tests~\cite{M-S-S,M-S}.  In this respect the equation resembles
the integrable Ernst equation of general relativity~\cite{B-M}.
\endgraf
Systematic computation of nonlocal symmetries soon reveals that infinitely many 
nonlocal symmetries can be obtained through commutation.  
This was first observed for the unreduced Benney system in~\cite{Odess-Sokolov 1}, 
where five symmetries were written out explicitly and the structure of the symmetry 
algebra was revealed. 
\endgraf
The aim of this paper is to provide an explicit description of symmetries of
the Gibbons--Tsarev equation and an exact proof that the symmetries constitute
the Witt algebra. As the reader will see, a rigorous proof is far from being
simple. 
Moreover, although the symmetry algebra has two generators
$Z^{(-1)}$ and $Z^{(1)}$ (constructed in
Sections~\ref{sec:local-symm-cons} and~\ref{sec:algebra-nonl-symm}, respectively),
they are not of much help, because obtaining them and their commutators requires
essentially the same effort as obtaining all symmetries and commutators at once.

\endgraf
We present the results as follows.  In Section~\ref{sec:prel-notat}, we
introduce main notions and the notation.  Section~\ref{sec:local-symm-cons}
deals with the local properties of the Gibbons--Tsarev equation. Coverings and
nonlocal conservation laws are dealt with in
Section~\ref{sec:infin-seri-nonl}. We introduce an appropriate infinite system
of nonlocal conservation laws in two different but equivalent ways, which is
convenient from the computational point of view. 
The corresponding infinite-dimensional covering is the common `ground' 
for all the nonlocal symmetries and their commutators to be constructed in the sequel.
We also consider a one-dimensional covering that allows to treat the equation as 
an evolutionary two-component system, which is also convenient for some proofs.  
In Section~\ref{sec:algebra-nonl-symm}, we construct the nonlocal symmetries,
starting with their shadows.  
The shadows were found in a way that can be reused
in other similar situations.  The rest of the section is devoted to the
explicit description of full nonlocal symmetries derived from these shadows
and to a proof that the symmetries constitute the Witt algebra.  
Finally,
Section~\ref{sec:uniq-symm} is devoted to the proof of uniqueness of the
constructed symmetries in the class of polynomial ones.

\section{Preliminaries and notation}
\label{sec:prel-notat}

We expose here briefly the fundamentals of local~\cite{AMS-book} and
nonlocal~\cite{Trends} geometry of PDEs.
Consider a PDE given by a system of relations $\{F=0\}$, where $F=
(F^1,\dots,F^r)$ is a vector function in $x = (x^1,\dots,x^n)$, $u =
(u^1,\dots,u^m)$ and finite number of partial derivatives of~$u$ with respect
to~$x$. To any such a system we put into correspondence a locus $\mathcal{E}
\subset J^\infty(\pi)$ in the space of infinite jets, where
$\pi\colon\mathbb{R}^m \times \mathbb{R}^n \to \mathbb{R}^n$ is the trivial
bundle and $\mathbb{R}^m$, $\mathbb{R}^n$ are Euclidean spaces with the
coordinates $u^1,\dots,u^m$, $x^1,\dots,x^n$, respectively. This locus is
defined by all the differential consequences of the system and called the
infinitely prolonged equation.

When the coordinates $x^i$, $u^j$ are chosen, the adapted coordinates
$u_\sigma^j$ arise in $J^\infty(\pi)$ and correspond to the partial
derivatives $\partial^{\abs{\sigma}}u^j/\partial x^\sigma$, where $\sigma$ is
a symmetric multi-index whose entries are the integers $1,\dots,n$. We always
assume that the system is presented in the passive orthonomic form,
see~\cite{Marvan-passive}, which allows to choose internal coordinates
on~$\mathcal{E}$. When we say that an object is restricted
from~$J^\infty(\pi)$ to~$\mathcal{E}$, we mean that it is rewritten in terms
of the internal coordinates.

The key role in the geometry of PDEs is played by the total derivative operators
\begin{equation*}
  D_{x^i} = \frac{\partial}{\partial x^i} + \sum_{j,\sigma}u_{\sigma i}^j
  \frac{\partial}{\partial u_\sigma^j}.
\end{equation*}
These operators can be restricted to any infinitely prolonged
equation. Consequently, any differential operator in total derivatives (a
$\mathcal{C}$-differential operator) is restrictable to~$\mathcal{E}$ as
well. We preserve the same notation for the restrictions if no contradiction
arises. We say that~$\mathcal{E}$ is differentially connected if
$D_{x^i}(f) = 0$, $i=1,\dots,n$ implies~$f=\const$. The distribution spanned
by the total derivatives is called the Cartan distribution and denoted
by~$\mathcal{C}$.

A vector field
\begin{equation*}
  S = \sum_{\Int} s_\sigma^j\frac{\partial}{\partial u_\sigma^j}
\end{equation*}
on $\mathcal{E}$ is a symmetry of $\mathcal{E}$ if $[S,D_{x^i}] = 0$ for all~$i$. The notation
$\sum_{\Int}$ means that the sum is taken over the set~$\Int$ of all internal
coordinates~$u_\sigma^j$. Symmetries form a Lie algebra denoted
by~$\sym(\mathcal{E})$. To describe symmetries, consider the following
construction. Let $G = (G^1,\dots,G^r)$ be a function on~$\mathcal{E}$. Define
its linearisation as the matrix $\mathcal{C}$-differential operator
\begin{equation*}
  \ell_G =
  \begin{pmatrix}
    \displaystyle\sum 
      \dfrac{\partial G^\alpha}{\partial
      u_\sigma^\beta}D_\sigma 
  \end{pmatrix}_{\beta=1,\dots,m}^{\alpha=1,\dots,r},
\end{equation*}
where $D_\sigma$ denotes the composition of~$D_{x^i}$ corresponding to the
multi-index~$\sigma$. We also use the notation~$\ell_{\mathcal{E}}$ for
$\eval{\ell_F}_{\mathcal{E}}$. Then the following result is valid: any
symmetry is an evolutionary vector field
\begin{equation*}
  \Ev_\phi = \sum_{\Int} D_\sigma(\phi^j)\frac{\partial}{\partial u_\sigma^j},
\end{equation*}
where the {\it generating section} $\phi = (\phi^1,\dots,\phi^m)$ satisfies the
equation~$\ell_{\mathcal{E}}(\phi) = 0$. The commutator of symmetries induces the
Jacobi bracket
\begin{equation*}
  \{\phi,\phi'\} = \Ev_{\phi}(\phi') - \Ev_{\phi'}(\phi).
\end{equation*}
We do not distinguish between symmetries and their generating sections
below. A symmetry~$S$ is called classical if it is projectable
to~$J^1(\pi)$. We say that~$S$ is a point symmetry if it is projectable
to~$J^0(\pi)$.

A conservation law of~$\mathcal{E}$ is a horizontal $(n-1)$-form
\begin{equation*}
  \omega = a_1\,dx^2\wedge\,dx^3\wedge\dots\wedge\,dx^n +
  a_2\,dx^1\wedge\,dx^3\wedge\dots\wedge\,dx^n +\dots +
  a_n\,dx^2\wedge\,dx^3\wedge\dots\wedge\,dx^{n-1}
\end{equation*}
closed with respect to the horizontal de~Rham differential $d_h = \sum_{i=1}^n
dx^i\wedge D_{x^i}$, i.e., such that $\sum_{i=1}^n(-1)^iD_{x^i}(a_i)=0$. A
conservation law is trivial if $\omega = d_h\rho$ for some
$(n-2)$-form~$\rho$. The quotient group of all conservation laws modulo
trivial ones is denoted by~$\Cl(\mathcal{E})$.

To compute conservation laws, their generating sections are used. Let $\omega$
be a conservation law and $\bar{\omega}$ be its extension to the ambient
space~$J^\infty(\pi)\supset\mathcal{E}$. Then $ d_h(\bar{\omega}) = \Delta(F)$
for some $\mathcal{C}$-differential operator~$\Delta$ and the vector-function
$ \psi = (\psi^1,\dots,\psi^r) = \eval{\Delta^*(1)}_{\mathcal{E}}$, where
$\Delta^*$ denotes the adjoint operator, is the generating section
of~$\omega$. It possesses two important properties: (a) $\psi=0$ if and only
if~$\omega$ is trivial, (b) $\ell_{\mathcal{E}}^*(\psi)=0$. Any solution of
the last equation is called a cosymmetry. The space of cosymmetries is denoted
by~$\cosym(\mathcal{E})$.

Let $\tilde{\mathcal{E}}$, $\mathcal{E}$ be equations. We say that a smooth
map $\tau\colon \tilde{\mathcal{E}} \to \mathcal{E}$ is a morphism if for any
point $\tilde{\theta} \in \tilde{\mathcal{E}}$ one has
$\tau_*(\tilde{\mathcal{C}}_{\tilde{\theta}}) \subset
\mathcal{C}_{\tau(\tilde{\theta})}$. A morphism is a (differential) covering
if $\eval{\tau_*}_{\tilde{\mathcal{C}}_{\tilde{\theta}}}$ is a isomorphism for
any~$\tilde{\theta} \in \tilde{\mathcal{E}}$. Two coverings~$\tau_1$, $\tau_2$
over $\mathcal{E}$ are equivalent if there exists an isomorphism $f\colon
\tilde{\mathcal{E}}_1 \to \tilde{\mathcal{E}}_2$ such that $\tau_2\circ f =
\tau_1$. Assume that $\mathcal{E}$ is differentially connected. Then we say
that~$\tau$ is irreducible if~$\tilde{\mathcal{E}}$ is differentially
connected as well.

Take coverings~$\tau_1$ and~$\tau_2$ and consider the Whitney product
$\tau_1\times \tau_2$ of the corresponding bundles. It carries a natural
structure of a covering, which is called the Whitney product of these
coverings.

Let $\tau\colon \tilde{\mathcal{E}} = \mathcal{E}\times\mathbb{R}^s$ be the
trivial bundle. Define the plane~$\tilde{\mathcal{C}}_{\tilde{\theta}}$ at a
point~$\tilde{\theta}\in\tilde{\mathcal{E}}$ as the parallel lift
of~$\mathcal{C}_\theta$ for~$\theta = \tau(\tilde{\theta})$. This is a
covering, and any covering is said to be trivial if it is equivalent
to~$\tau$.

A one-dimensional covering~$\tau$ is called Abelian if it is either trivial or
there exists a nontrivial conservation law~$\omega$ of~$\mathcal{E}$ such that
its lift~$\tau^*(\omega)$ becomes trivial on~$\tilde{\mathcal{E}}$. In
general, a covering is Abelian if it is equivalent to the Whitney product of
the necessary number of one-dimensional Abelian coverings. 
\begin{proposition}[see~\cite{IiC}]\label{sec:prel-notat-1}
  A finite-dimensional Abelian covering over~$\mathcal{E}$ is irreducible if
  and only if the corresponding system of conservation laws is linearly
  independent modulo trivial ones. Consequently\textup{,} equivalence classes
  of irreducible $s$-dimensional\textup{,} $s<\infty$\textup{,} Abelian
  coverings are in one-to-one correspondence with $s$-dimensional subspaces
  in~$\Cl(\mathcal{E})$.
\end{proposition}

We say that a symmetry of the covering equation~$\tilde{\mathcal{E}}$ is a
nonlocal symmetry of~$\mathcal{E}$. Denote by~$\mathcal{F}$ and
$\tilde{\mathcal{F}}$ the algebras of smooth functions on~$\mathcal{E}$
and~$\tilde{\mathcal{E}}$. Due to~$\tau$, one has the
embedding~$\mathcal{F} \subset \tilde{\mathcal{F}}$. A derivation
$\mathcal{F} \to \tilde{\mathcal{F}}$ that preserves the Cartan distributions
is called a nonlocal shadow. In particular, for any nonlocal
symmetry~$\tilde{S}$ its restriction $\eval{\smash{\tilde{S}}}_{\mathcal{F}}$
is a shadow; $\tilde{S}$ is said to be {\it invisible} if its shadow vanishes. Local
symmetries of~$\mathcal{E}$ can be treated as shadows in every covering. A
shadow is called reconstructible if there exists a nonlocal symmetry such that
its shadow is the given one.

We also say that a conservation law of~$\tilde{\mathcal{E}}$ is a nonlocal
conservation law of~$\mathcal{E}$.

Let us pass to local coordinates. Since the Gibbons--Tsarev equation is
two-dimensional, we shall confine ourselves to this case for
simplicity. Consider the equation $\mathcal{E}$ given by
$\{F(x,y,u,u_x,u_y,\dots) = 0\}$ and let $\tau\colon \tilde{\mathcal{E}} =
\mathcal{E} \times \mathbb{R}^s \to \mathcal{E}$ be a covering (the case
$s=\infty$ is allowed). Let $\{w^\alpha\}$ be coordinates in the fiber (they
are called nonlocal variables). Then the total derivatives
on~$\tilde{\mathcal{E}}$ are of the form
\begin{equation*}
  \tilde{D}_x = D_x + \sum_\alpha X_\alpha\frac{\partial}{\partial
    w^\alpha},\qquad \tilde{D}_y = D_y + \sum_\alpha
  Y_\alpha\frac{\partial}{\partial 
    w^\alpha},  
\end{equation*}
$X_\alpha$, $Y_\alpha$ being smooth functions in all the internal variables
and~$w^\alpha$. Then $\tau$ is a covering if and only if
$[\tilde{D}_x,\tilde{D}_y] = 0$, or, equivalently,
\begin{equation*}
  D_x(Y_\alpha) - D_y(X_\alpha) + \sum_\beta
  \left(X_\beta\frac{\partial Y_\alpha}{\partial w^\beta} -
    Y_\beta\frac{\partial X_\alpha}{\partial w^\beta}\right) = 0,\quad \alpha
  = 1,\dots,s.
\end{equation*}
Equivalently, the system
\begin{equation}
  \label{eq:2}
  w_x^\alpha = X_\alpha,\qquad w_y^\alpha = Y_\alpha,
\end{equation}
is compatible modulo~$\mathcal{E}$. If the
functions $X_\alpha$, $Y_\alpha$ do not depend on the nonlocal variables, then
the covering is Abelian.

Any nonlocal symmetry in~$\tau$ is defined by its generating section
$\Phi = (\phi,\dots,\psi^\alpha,\dots)$, where $\phi = (\phi^1,\dots,\phi^m)$
and $\psi^\alpha$ are functions on $\tilde{\mathcal{E}}$ satisfying
\begin{gather}\nonumber
  \tilde{D}_x(\psi^\alpha) = \tilde{\ell}_{X_\alpha}(\phi) +
  \sum_\beta\frac{\partial X_\alpha}{\partial w^\beta}\psi^\beta,\quad
  \tilde{D}_y(\psi^\alpha) = \tilde{\ell}_{Y_\alpha}(\phi) +
  \sum_\beta\frac{\partial Y_\alpha}{\partial w^\beta}\psi^\beta,
  \\\label{eq:4}
  \tilde{\ell}_{\mathcal{E}}(\phi) = 0,
\end{gather}
where the `tilde' over a $\mathcal{C}$-differential operator denotes its natural
lift from~$\mathcal{E}$ to~$\tilde{\mathcal{E}}$. Nonlocal shadows are given
by functions~$\phi$ that satisfy Equation~\eqref{eq:4}, while invisible
symmetries are sections~$\Phi$ with $\phi=0$ and $\psi^\alpha$ satisfying
\begin{equation*}
  \tilde{D}_x(\psi^\alpha) = 
  \sum_\beta\frac{\partial X_\alpha}{\partial w^\beta}\psi^\beta,\qquad
  \tilde{D}_y(\psi^\alpha) = 
  \sum_\beta\frac{\partial Y_\alpha}{\partial w^\beta}\psi^\beta.
\end{equation*}

Assume that the right-hand sides of~\eqref{eq:2} depend on a
parameter~$\lambda$ (which is called the spectral parameter). A parameter is
non-removable (essential) if the coverings $\tau_\lambda$ are pair-wise
inequivalent (cf.~\cite{Marvan-non-remove}). Having a family of coverings with
an essential parameter, one can expand the functions~$X_\alpha$, $Y_\alpha$ in
formal series in~$\lambda$. If substitution
of~$\psi^\alpha = \sum_{i\in \mathbb{Z}} \psi_i^\alpha \lambda^i$ to this
expansion is well defined, one obtains an infinite-dimensional covering with
the nonlocal variables~$\psi_i^\alpha$. In the case when this covering is
Abelian we get an infinite family of conservation laws (perhaps, trivial or
dependent). A classical example of this procedure is the construction of the
infinite series of conservation laws for the Korteweg--de~Vries
equation,~\cite{Miura-Gardner}.

But even if the covering at hand does not depend on a parameter, there exists
a standard way to insert such a parameter formally (the so-called reversion
procedure, see~\cite{Pavlov-reversion}). Namely, assume for simplicity
that~$\tau$ is one-dimensional and is given by
\begin{equation*}
  w_x = X(x,y,w,u,u_x,u_y,\dots),\qquad w_y = Y(x,y,w,u,u_x,u_y,\dots).
\end{equation*}
Then
\begin{equation*}
  v_x = -X(x,y,\lambda,u,u_x,u_y,\dots)v_\lambda,\qquad v_y =
  -Y(x,y,\lambda,u,u_x,u_y,\dots)v_\lambda 
\end{equation*}
is a covering as well, see~\cite{Krasil-3G} for the geometric
interpretation. We use this construction below to construct infinite series of
nonlocal conservation laws for the Gibbons--Tsarev equation.

\section{Local symmetries and conservation laws}
\label{sec:local-symm-cons}

Consider the Gibbons--Tsarev equation~\cite{G-T 1} in the form
\begin{equation}\label{eq:21}
  z_{yy} + z_x z_{xy} - z_y z_{xx} + 1 = 0
\end{equation}
(obtained from~\cite[eq.~(15)]{G-T 1} by the exchange 
$x \leftrightarrow y$ and $z \leftrightarrow -z$).
For a monomial $X = x^i y^j$, let us use the notation
\begin{equation*}
  z_X = \frac{\partial^{i+j} z}{\partial x^i\,\partial y^j}.
\end{equation*}
In particular, $z_X = z$ when $i = j = 0$. For internal coordinates
on~$\mathcal{E}$ we choose 
$x$, $y$, $z_X$ such that $X = x^k$
or $X = x^k y$, $k\geq0$, while
$$
z_{yyX} = D_{X}(z_y z_{xx} - z_x z_{xy} - 1)
$$ 
are functions of the internal coordinates for every monomial $X$.

If not stated otherwise, sums are taken over all internal coordinates.
The total derivatives on~$\mathcal{E}$ are
\begin{equation*}
  D_x = \frac{\partial}{\partial x}
  + \sum z_{xX} \frac{\partial}{\partial z_X}, \qquad
  D_y = \frac{\partial}{\partial y}
  + \sum z_{yX} \frac{\partial}{\partial z_X},
\end{equation*}
(summation over all internal coordinates $z_X$).  It is
straightforward to check that that~\eqref{eq:21} is a differentially connected
equation.

\subsection{Weights}
\label{sec:weights}

The Gibbons--Tsarev equation becomes homogeneous if we assign the weights 
$\wt{x} = 3$, $\wt{y} = 2$, $\wt{z} = 4$ 
(due to the scaling symmetry, see Subsection~\ref{sec:point-symmetries} below)  
and
\begin{equation*}
  \wt{z_{x^k}}=\wt{z}-k\wt{x}=4-3k,\quad\wt{z_{x^ky}} =
  \wt{z}-k\wt{x}-\wt{y} = 2-3k.
\end{equation*}
To any monomial in~$x$, $y$, $z_{x^k}$, and~$z_{x^ky}$ we assign the weight
that equals the sum of weights of its factors. The total derivatives preserve
the space of polynomials and, as operators, have the weights $
  \wt{D_x}=-3$, $\wt{D_y}=-2$.

\subsection{Local symmetries}\label{sec:point-symmetries}

Let~$\mathcal{S}=\Ev_Z$ be a symmetry of~$\mathcal{E}$. Then the defining
equation for the generating sections of symmetries is
\begin{equation}
  \label{eq:22}
  \ell_{\mathcal{E}}(Z)\equiv D_y^2(Z) + z_x D_xD_y(Z) -
  z_y D_x^2(Z) + z_{xy} D_x(Z) - z_{xx} D_y(Z) = 0.
\end{equation}

Solving~\eqref{eq:22} for functions~$Z$ of small jet order, we found that
Equation~\eqref{eq:21} possesses five local symmetries 
\begin{align}\nonumber
    \Zi{-4} &= 1,&&z\text{-translation},\\\nonumber
    \Zi{-3} &= z_x,&&x\text{-translation},\\\label{eq:23}
    \Zi{-2} &= z_y,&&y\text{-translation},\\\nonumber
    \Zi{-1} &= y z_x - 2 x,&&\text{generalized Galilean boost},\\\nonumber
    \Zi{0} &= 3 x z_x + 2 y z_y - 4 z,&&\text{scaling}.
\end{align}
All these symmetries are point ones. In Section~\ref{sec:uniq-symm}, it will
be shown that this is the complete set of local symmetries. The vector
field~$\Si{i}=\Ev_{\Zi{i}}$, as an operator, has the weight~$\wt{\Si{i}}=i$.

All commutators of the symmetries $\Si{-4},\dots,\Si{0}$ vanish except for
\begin{align*}
  &[\Si{0}, \Si{-4}] = 2 \Si{-4},&& \quad [\Si{0}, \Si{-3}] = \tfrac32
  \Si{-3},\\
  &[\Si{0}, \Si{-2}] = \Si{-2}, &&\quad [\Si{0}, \Si{-1}] =\tfrac12 \Si{-1},
  \\
  &[\Si{-1}, \Si{-3}] = -2 \Si{-4}, &&\quad [\Si{-1}, \Si{-2}] = -\Si{-3}.
\end{align*}

\begin{remark}
  Note that changing the basis by~$\Si{0}\mapsto -\frac{1}{2}\Si{0}$ we arrive
  to the commutator relations $ [\Si{i},\Si{j}]=(j-i)\Si{i+j}$, where
  formally~$\Si{\alpha}=0$ for~$\alpha<-4$. In what follows, we use the latter
  choice of the basic symmetries.
\end{remark}

It will be shown in Section~\ref{sec:algebra-nonl-symm} that this set of five
symmetries can be extended to a hierarchy of nonlocal symmetries infinite in
both positive and negative directions.

\subsection{Cosymmetries}
\label{sec:cosymmetries}
The defining equation for cosymmetries of~\eqref{eq:21} is
\begin{equation*}
  \ell_{\mathcal{E}}^*(\mathcal{R})\equiv D_y^2(\mathcal{R}) +
  z_x D_xD_y(\mathcal{R}) - z_y D_x^2(\mathcal{R}) -
  2z_{xy} D_x(\mathcal{R}) + 2z_{xx} D_x(\mathcal{R})=0.
\end{equation*}
Solutions of lower order include six local cosymmetries of the first order
\begin{align*}
  \Ri{0} &= 1,\\
  \Ri{1} &= 2 z_x,\\
  \Ri{2} &= 3 z_x^2 + 2 z_y + 3 y,\\
  \Ri{3} &= 4 z_x^3 + 6 z_x z_y + 8 y z_x + 2 x,\\
  \Ri{4} &= 5 z_x^4 + 12 z_x^2 z_y + 15 y z_x^2 + 3 z_y^2 + 6 x z_x
  + 10 y z_y + z + \tfrac{15}{2} y^2,\\
  \Ri{5} &= 6 z_x^5 + 20 z_x^3 z_y + 24 y z_x^3 + 12 z_x z_y^2 + 12 x z_x^2 +
  36
  y z_x z_y \\
  &\qquad + 4 (z + 6 y^2) z_x + 8 x z_y + 12 x y
  \intertext{(compare
    with~\cite[p.~156]{Tsarev-PhD}) and a single one of the third order}
  \Ri{-5} &= z_{xxx}.
\end{align*}
We have verified by direct computation that Equation~\eqref{eq:21} has no
local generating section of order~$2$ and~$4$.

\subsection{Conservation laws}\label{sec:conservation-laws}
All the above listed cosymmetries are the generating sections of conservation
laws $ \rhoi{i}=\P{i}\diff x +\Q{i}\diff y$, where
\begin{align*}
  \P0 &= z_x^2 + z_y + y, \\
  \Q0 &= z_x z_y, \\[5pt]
  \P1 &= z_x^3 + 2 z_x z_y - x, \\
  \Q1 &= z_x^2 z_y + z_y^2 - 2 z, \\[5pt]
  \P2 &= z_x^4 + 3 z_x^2 z_y + 3 y z_x^2 + z_y^2 + 3 y z_y - z, \\
  \Q2 &= z_x^3 z_y + 2 z_x z_y^2 + 3 y z_x z_y - 3 x y, \\[5pt]
  \P3 &= z_x^5 + 4 z_x^3 z_y + 4 y z_x^3 + 3 z_x z_y^2 + 2 x z_x^2 + 8 y z_x
  z_y - 2 z z_x + 2 x z_y - 4 x y, \\ 
  \Q3 &= z_x^4 z_y + 3 z_x^2 z_y^2 + 4 y z_x^2 z_y + z_y^3 + 2 x z_x z_y + 4 y
  z_y^2 - 2 z z_y - 8 y z - 3 x^2,
  \\[5pt]
  \P4 &= z_x^6 + 5 z_x^4 z_y + 5 y z_x^4 + 6 z_x^2 z_y^2 + 3 x z_x^3 + 15 y
  z_x^2 z_y + z_y^3 + (z + \tfrac{15}{2} y^2) z_x^2 \\&\quad + 6 x z_x z_y + 5
  y z_y^2 + (z + \tfrac{15}{2} y^2) z_y - 5 y z - 4 x^2,
  \\
  \Q4 &= z_x^5 z_y + 4 z_x^3 z_y^2 + 5 y z_x^3 z_y + 3 z_x z_y^3 + 3 x z_x^2
  z_y + 10 y z_x z_y^2
  \\&\quad + (z + \tfrac{15}{2} y^2) z_x z_y + 3 x z_y^2 - \tfrac{3}{2} x (4 z
  + 5 y^2), \\[5pt] 
  \P5 &= z_x^7 + 6 z_x^5 z_y + 6 y z_x^5 + 10 z_x^3 z_y^2 + 4 x z_x^4 + 24 y
  z_x^3 z_y + 4 z_x z_y^3 + 4 (\tfrac{1}{2} z + 3 y^2) z_x^3 \\&\quad + 12 x
  z_x^2 z_y + 18 y z_x z_y^2 + 12 x y z_x^2 + 4 (z + 6 y^2) z_x z_y + 4 x
  z_y^2 + 12 x y z_y - 4 x z, \\
  \Q5 &= z_x^6 z_y + 5 z_x^4 z_y^2 + 6 y z_x^4 z_y + 6 z_x^2 z_y^3 + 4 x z_x^3
  z_y + 18 y z_x^2 z_y^2 + z_y^4 + 4 (\tfrac{1}{2} z + 3 y^2) z_x^2 z_y
  \\&\quad + 8 x z_x z_y^2 + 6 y z_y^3 + 12 x y z_x z_y + 4 (\tfrac{1}{2} z +
  3 y^2) z_y^2 - 2 z^2 - 24 y^2 z - 6 x^2 y
  \intertext{and}
  \P{-5} &= 
      -\tfrac{1}{2} z_x z_{xx}^2 - z_{xy} z_{xx},\\
  \Q{-5} &= 
      -\tfrac{1}{2} z_y z_{xx}^2 - \tfrac{1}{2} z_{xy}^2.
\end{align*}
Note that $ \wt{\rhoi{i}}=i+5$, $i=-5,0,\dots,5$.

In the next section we construct an infinite series of \emph{nonlocal}
conservation laws for the Gibbons--Tsarev equation~\eqref{eq:21}.

\section{Coverings and the infinite series of nonlocal conservation laws}
\label{sec:infin-seri-nonl}

Using two known coverings~\cite{G-T 1,G-T 2} of the Gibbons--Tsarev equation,
we construct here an infinite series of (nonlocal) conservation laws that later
(Section~\ref{sec:algebra-nonl-symm}) will be used to construct the
corresponding infinite-dimensional Abelian covering and describe the algebra
of nonlocal symmetries in this covering. It will also be shown that the
obtained infinite dimensional coverings are equivalent.

\subsection{Coverings}
\label{sec:coverings}

Consider the nonlinear non-Abelian
covering~$\tau_z\colon\tilde{\mathcal{E}}\to \mathcal{E}$ over
Equation~\eqref{eq:21} given by
\begin{equation}
  \label{eq:26}
  \phi_x=\frac{1}{z_y+z_x\phi-\phi^2},\qquad
  \phi_y=-\frac{z_x-\phi}{z_y+z_x\phi-\phi^2}.
\end{equation}
The covering introduced by Gibbons and Tsarev in~\cite{G-T 2} can be rewritten
in this way.

To simplify the subsequent computations, let us introduce new variables~$u$
and~$v$ such that 
\begin{equation}
  \label{eq:27}
  z_x=u+v,\qquad z_y=-uv.
\end{equation}
Due to the compatibility condition
\begin{equation}
  \label{eq:28}
  (u+v)_y+(uv)_x=0
\end{equation}
and by Equation~\eqref{eq:21} we deduce that the new variables enjoy the
system of evolution equations
\begin{equation}
  \label{eq:29}
  u_y+vu_x=\frac{1}{v-u},\qquad v_y+uv_x=\frac{1}{u-v}
\end{equation}
Denote this equation by~$\mathcal{E}_1$. The equation is homogeneous with
respect to the weights $ \wt{x}=3$, $\wt{y}=2$, $\wt{u}=\wt{v}=1$.  Due
to~\eqref{eq:28}, the form~$(u+v)\diff x-uv\diff y$ is a conservation law of
the equation~$\mathcal{E}_1$ while~\eqref{eq:27} defines the
covering~$\mathcal{E}\to\mathcal{E}_1$ associated with this conservation law.

The covering~$\tau_z$ defined by~\eqref{eq:26} generates the
covering~$\tau_{uv}\colon\tilde{\mathcal{E}}_1\to \mathcal{E}_1$ given by the
relations
\begin{equation}
  \label{eq:30}
  \phi_x=-\frac{1}{(\phi-u)(\phi-v)},\qquad
  \phi_y=\frac{u+v-\phi}{(\phi-u)(\phi-v)} 
\end{equation}
and the diagram of coverings
\begin{equation*}
  \xymatrix{
    \tilde{\mathcal{E}}\ar[d]_{\tau_z}\ar[r]^{\tilde{\tau}}&
    \tilde{\mathcal{E}}_1\ar[d]^{\tau_{uv}}\\ 
    \mathcal{E}\ar[r]^\tau&\mathcal{E}_1
  }
\end{equation*}
is commutative.

\subsection{Nonlocal conservation laws}
\label{sec:nonl-cons-laws}

We construct an infinite hierarchy of nonlocal conservation laws for the
Gibbons-Tsarev equation using two different but related ways.

\subsubsection{The first way}
\label{sec:first-method}

Consider an arbitrary gauge symmetry~$\phi\mapsto\psi(\phi)$ of the
covering~$\tau_{uv}$. For the sake of convenience, relabel the variable~$\phi$
to~$\lambda$. Then, applying the reversion procedure described in
Section~\ref{sec:prel-notat} to the covering~\eqref{eq:30}, one obtains
\begin{equation}
  \label{eq:31}
  \psi_x=\frac{1}{(\lambda-u)(\lambda-v)}\cdot\psi_\lambda,\qquad
  \psi_y=\frac{\lambda-(u+v)}{(\lambda-u)(\lambda-v)}\cdot\psi_\lambda.
\end{equation}
Now, we consider~$\lambda$ as a formal parameter and expand~$\psi$ in the
Laurent series
\begin{equation}
  \label{eq:32}
  \psi=\psii{-1}\lambda+\psii{0}+\frac{\psii{1}}{\lambda}+ \dots +
  \frac{\psii{k}}{\lambda^k} + \dots
\end{equation}
One also has the obvious expansions
\begin{equation*}
  \frac{1}{\lambda-u}=\frac{1}{\lambda}\sum_{i\geq0}\frac{u^i}{\lambda^i},
  \qquad
  \frac{1}{\lambda-v}=\frac{1}{\lambda}\sum_{i\geq0}\frac{v^i}{\lambda^i}
\end{equation*}
which imply
\begin{equation*}
  \frac{1}{(\lambda-u)(\lambda-v)}=\frac{1}{\lambda^2}\left(1 +
    \frac{\sigma_1}{\lambda} + \dots+\frac{\sigma_k}{\lambda^k}+\dots\right),
\end{equation*}
where
\begin{equation}
\label{sigma by uv}
  \sigma_k=\sum_{i+j=k}u^iv^j.
\end{equation}

\begin{remark}
  Note that since the quantities~$\sigma_k$ are symmetric in the variables~$u$
  and~$v$, they can be rewritten as polynomials in $z_x=u+v$
  and~$z_y=-uv$. See formula~\eqref{sigma by z} below.
\end{remark}

Now, from the expansion~\eqref{eq:32} one obtains
\begin{align*}
  \psi_x&=\psii{-1}_x\lambda+\psii{0}_x+\frac{\psii{1}_x}{\lambda}+ \dots +
  \frac{\psii{k}_x}{\lambda^k} + \dots,\\
  \psi_y&=\psii{-1}_y\lambda+\psii{0}_y+\frac{\psii{1}_y}{\lambda}+ \dots +
  \frac{\psii{k}_y}{\lambda^k} + \dots,
  \intertext{and}
  \psi_\lambda&=\psii{-1}-\frac{\psii{1}}{\lambda^2} -
  2\frac{\psii{2}}{\lambda^3} -\dots - k\frac{\psii{k}}{\lambda^{k+1}}+\dots
\end{align*}
Substituting all the above expansions to Equations~\eqref{eq:31}, one obtains
\begin{align*}
  &\psii{-1}_x\lambda+\psii{0}_x+\frac{\psii{1}_x}{\lambda}+ \dots +
  \frac{\psii{k}_x}{\lambda^k} + \dots\\
  &\ =\frac{1}{\lambda^2}\left(1 + \frac{\sigma_1}{\lambda} +
    \dots+\frac{\sigma_k}{\lambda^k}+\dots\right)
  \left(\psii{-1}-\frac{\psii{1}}{\lambda^2} - \frac{2\psii{2}}{\lambda^3}
    -\dots -
    \frac{k\psii{k}}{\lambda^{k+1}}+\dots\right),\\
  &\psii{-1}_y\lambda+\psii{0}_y+\frac{\psii{1}_y}{\lambda}+ \dots +
  \frac{\psii{k}_y}{\lambda^k} + \dots\\
  &\ =\left(\frac{1}{\lambda}-\frac{\sigma_1}{\lambda^2}\right) \left(1 +
    \frac{\sigma_1}{\lambda} + \dots+\frac{\sigma_k}{\lambda^k}+\dots\right)
  \left(\psii{-1}-\frac{\psii{1}}{\lambda^2} - \frac{2\psii{2}}{\lambda^3}
    -\dots - \frac{k\psii{k}}{\lambda^{k+1}}+\dots\right).
\end{align*}
Denote by
\begin{equation*}
  A_0 + \frac{A_1}{\lambda} +\dots+ \frac{A_k}{\lambda^k}+\dots
\end{equation*}
the result of multiplication of the last two factors in the previous
expressions, i.e.,~$A_0=\psii{-1}$, $A_1=\sigma_1\psii{-1}$,
$A_2=\sigma_2\psii{-1}-\psii{1}$, and
\begin{equation*}
  A_k=\sigma_k\psii{-1}-\sigma_{k-2}\psii{1} - 2\sigma_{k-3}\psii{2} - \dots
  -(k-2)\sigma_1\psii{k-2} -(k-1)\psii{k-1},\quad k\geq3.
\end{equation*}
Consequently,
\begin{align*}
  &\psii{-1}_x\lambda+\psii{0}_x+\frac{\psii{1}_x}{\lambda}+ \dots +
  \frac{\psii{k}_x}{\lambda^k} + \dots=\frac{1}{\lambda^2}\left(A_0 +
    \frac{A_1}{\lambda} +\dots+ \frac{A_k}{\lambda^k}+\dots\right),
  \\
  &\psii{-1}_y\lambda+\psii{0}_y+\frac{\psii{1}_y}{\lambda}+ \dots +
  \frac{\psii{k}_y}{\lambda^k} + \dots=
  \left(\frac{1}{\lambda}-\frac{\sigma_1}{\lambda^2}\right)\left(A_0 +
    \frac{A_1}{\lambda} +\dots+ \frac{A_k}{\lambda^k}+\dots\right)
\end{align*}
and thus
\begin{align}\nonumber
  &\psii{-1}_x=0,\ \psii{0}_x=0,\ \psii{1}_x=0,
  &&\psii{-1}_y=0,\ \psii{0}_y=0,\ \psii{1}_y=A_1 \\
  \intertext{and}
  &\psii{k}_x=A_{k-2}, 
  &&\psii{k}_y=A_{k-1}-\sigma_1 A_{k-2}\label{eq:33}
\end{align}
for~$k\geq2$. Without loss of generality we can set~$\psii{-1}=1$ and skip the
variable~$\psii{0}$, since the coefficients~$A_k$ are independent of it. Then,
using the obtained expressions for~$A_0$ and~$A_1$, we obtain $ \psii{1}_x=0$,
$\psii{1}_y=1$, $\psii{2}_x=1$, $\psii{2}_y=0$ and set
\begin{equation}\label{eq:37}
  \psii{1}=y,\quad\psii{2}=x,
\end{equation}
without loss of generality as well. Thus, we have
\begin{equation*}
  A_0=1,\ A_1=\sigma_1,\ A_2=\sigma_2-y,\ A_3=\sigma_3-\sigma_1y-2x
\end{equation*}
and
\begin{equation*}
  A_k=\sigma_k-\sigma_{k-2}y-2\sigma_{k-3}x-3\sigma_{k-4}\psii{3}- \dots -
  (k-2)\sigma_1\psii{k-2} -(k-1)\psii{k-1}
\end{equation*}
for~$k>3$.

Then, using the obvious identities
$ \sigma_1\sigma_k-\sigma_{k+1}=uv\sigma_{k-1}$, we obtain from~\eqref{eq:33}
\begin{align*}
  &\psii{3}_x=\sigma_1,&&\psii{3}_y=-uv-y;\\
  &\psii{4}_x=\sigma_2-y,&&\psii{4}_y=-uv\sigma_1-2x;\\
  &\psii{5}_x=\sigma_3-\sigma_1 y-2x,&&\psii{5}_y=-uv(\sigma_2 -
  y)-3\psii{3};\\
  &\psii{6}_x=\sigma_4-\sigma_2 y-2\sigma_1 x -
  3\psii{3},&&\psii{6}_y=-uv(\sigma_3 - \sigma_1 y - 2x) -4\psii{4};\\
  &\psii{7}_x=\sigma_5 - \sigma_3 y - 2\sigma_2 x -3\sigma_1\psii{3} -
  4\psii{4}, &&\psii{7}_y=-uv(\sigma_4 - \sigma_2 y - 2\sigma_1 x - 3\psii{3})
  - 5\psii{5}
\end{align*}
and
\begin{align}
  \label{eq:69}
  \psii{k}_x&=\sigma_{k-2}-\sigma_{k-4}y - 2\sigma_{k-3}x -
  \sum_{i=3}^{k-3}i\sigma_{k-i-3} \psii{i},\\\label{eq:70}
  \psii{k}_y&=-uv(\sigma_{k-3}-\sigma_{k-5}y - 2\sigma_{k-6}x -
  \sum_{i=3}^{k-4}i\sigma_{k-i-4}\psii{i}) - (k-2)\psii{k-2}.
\end{align}
for~$k\geq7$. Denote by~$\Xii{k}$ and~$\Yii{k}$ the right-hand sides of the
obtained equations, i.e.,
\begin{equation}
  \label{XY}
  \psii{k}_x=\Xii{k},\quad\psii{k}_y=\Yii{k},\qquad k\geq3.
\end{equation}
Obviously, we have $ \wt{\Xii{k}}=k-2$, $\wt{\Yii{k}}=k-1$,
$\wt{\psii{k}}=k+1$.

Let us now return back to the equation~$\mathcal{E}_1$ given by~\eqref{eq:29}
and consider the spaces
\begin{equation*}
  \mathcal{E}_2=\mathcal{E}_1\times \Rii{3},\dots,\mathcal{E}_k=
  \mathcal{E}_{k-1}\times \Rii{k+1},\dots,
\end{equation*}
where~$\Rii{k}$ is~$\mathbb{R}^1$ with the distinguished coordinate~$\psii{k}$,
$k\geq3$. Consider also the natural projections
\begin{equation*}
  \tau_{k,k-1}\colon\mathcal{E}_{k}\to\mathcal{E}_{k-1},\qquad
  \tau_k\colon\mathcal{E}_k\to\mathcal{E}_1.
\end{equation*}
Let~$\mathcal{E}_*$ be the inverse limit of the infinite sequence
\begin{equation*}
  \xymatrix{
    \mathcal{E}_1&\ar[l]\dots&\ar[l]\mathcal{E}_{k-1}
    &\ar[l]_-{\tau_{k,k-1}}\mathcal{E}_k&\ar[l]\dots 
  }
\end{equation*}
and~$\tau_*\colon\mathcal{E}_*\to\mathcal{E}_1$ be the corresponding
projection. Endow the spaces~$\mathcal{E}_k$ with the vector fields
\begin{equation*}
  \Dii{k}_x=D_x+\sum_{i=3}^{k+1}\Xii{i}\frac{\partial}{\partial\psi{i}},\quad
  \Dii{k}_y=D_y+\sum_{i=3}^{k+1}\Yii{i}\frac{\partial}{\partial\psi{i}},
\end{equation*}
where~$D_x$ and~$D_y$ are the total derivatives on~$\mathcal{E}_1$. Similarly,
we define the fields~$\Dii{*}_x$ and~$\Dii{*}_y$ on~$\mathcal{E}_*$.

\begin{proposition}\label{sec:first-method-1}
  For all~$k$\textup{,} including the case~$k=*$\textup{,} one has
  $ [\Dii{k}_x,\Dii{k}_y]=0$.
\end{proposition}

\begin{proof}
  This is an immediate consequence of the fact that~\eqref{eq:30} is a
  covering over~$\mathcal{E}_1$.
\end{proof}

Hence, all the maps~$\tau_k$ carry covering structures; these coverings are
irreducible:

\begin{proposition}\label{sec:first-method-2}
  Let~$f\in\mathcal{F}(\mathcal{E}_k)$ be a function such that
  $ \Dii{k}_x(f)=\Dii{k}_y(f)=0$.  Then~$f=\const$.
\end{proposition}

\begin{proof}
  Let
  \begin{equation*}
    x,\ y,\dots, u_i=\frac{\partial^iu}{\partial
      x^i},\ v_i=\frac{\partial^iv}{\partial x^i},\dots 
  \end{equation*}
  be coordinates on~$\mathcal{E}_1$ and
  \begin{align*}
    D_x&=\frac{\partial}{\partial x} +
    \sum_{i\geq0}\left(u_{i+1}\frac{\partial}{\partial u_i} +
      v_{i+1}\frac{\partial}{\partial v_i}\right),\\
    D_y&=\frac{\partial}{\partial y} + \sum_{i\geq0}\left(
      D_x^i\left(\frac{1}{v-u} + vu_1\right)\frac{\partial}{\partial u_i} +
      D_x^i\left(\frac{1}{u-v} + uv_1\right)\frac{\partial}{\partial
        v_i}\right)
  \end{align*}
  be the total derivatives in these coordinates. Consider a function
  \begin{equation*}
    f=f(x,y,u,v,\dots,u_{i},v_{j},\psii{3},\dots,\psii{k})
  \end{equation*}
  on~$\mathcal{E}_k$ and assume that
  \begin{equation}\label{eq:36}
    D_x(f) + \Xii{3}\frac{\partial f}{\partial\psii{3}} +\dots+
    \Xii{k}\frac{\partial f}{\partial\psii{k}} = D_y(f) +
    \Yii{3}\frac{\partial f}{\partial\psii{3}} + \dots+
    \Yii{k}\frac{\partial f}{\partial\psii{k}} = 0.
  \end{equation}
  Since the coefficients~$\Xii{3}$, $\Yii{3},\dots,\Xii{k}$, $\Yii{k}$ are
  independent of the variables~$u_\alpha$, $v_\beta$ for all $\alpha$ and
  $\beta>0$, from the above formulas for~$D_x$ an~$D_y$ it follows that~$f$
  cannot depend on these variables either as well as on~$u$ and~$v$ and thus
  Equation~\eqref{eq:36} reads now
  \begin{equation*}
    \frac{\partial f}{\partial x} + \Xii{3}\frac{\partial f}{\partial\psii{3}}
    + \dots+
    \Xii{k}\frac{\partial f}{\partial\psii{k}} = \frac{\partial f}{\partial
      y}+ 
    \Yii{3}\frac{\partial f}{\partial\psii{3}} + \dots+
    \Yii{k}\frac{\partial f}{\partial\psii{k}} = 0.
  \end{equation*}
  But~$\Xii{\alpha}$ and~$\Yii{\beta}$ are polynomials in~$u$ and~$v$ of
  degrees~$\alpha-2$ and~$\beta-1$, respectively, and this finishes the proof.
\end{proof}

Obviously, every map~$\tau_{k,k-1}\colon\mathcal{E}_k\to\mathcal{E}_{k-1}$ is
also a covering; moreover, it is an Abelian covering associated to the
conservation law
\begin{equation*}
  \omegai{k}=\Xii{k}\diff x+\Yii{k}\diff y\in\Cl(\mathcal{E}_{k-1})
\end{equation*}
and~$\wt{\omegai{k}}=k+1$.

\begin{proposition}
  The conservation law~$\omegai{k}$ is nontrivial on~$\mathcal{E}_{k-1}$.
\end{proposition}

\begin{proof}
  This readily follows from general properties of coverings (see
  Section~\ref{sec:prel-notat}) and Proposition~\ref{sec:first-method-2}.
\end{proof}

\begin{remark}
  By the very construction, the equation~$\mathcal{E}_2$ is equivalent to the
  Gibbons--Tsarev equation~\eqref{eq:21}. Moreover, it can be checked that the
  conservation laws~$\omegai{4},\dots,\omegai{9}$ are equivalent to the
  conservation laws~$\rhoi{0},\dots,\rhoi{5}$, respectively, described in
  Subsection~\ref{sec:conservation-laws}.
\end{remark}

\begin{remark}
  Of course, the initial choice~\eqref{eq:37} for the values of~$\psii{-1}$,
  $\psii{1}$, and~$\psii{2}$ is not unique. Nevertheless, one can easily show
  that other admissible values lead to equivalent results.
\end{remark}

\subsubsection{The second method}
\label{sec:second-method}

Consider now the covering~\eqref{eq:30} and assume that
\begin{equation}\label{eq:39}
  \phi=\frac{\phii{-1}}{\lambda} + \phii{0} + \phii{1}\lambda + \dots +
  \phii{k}\lambda^k + \dots
\end{equation}
Then, rewriting~\eqref{eq:31} in the form
\begin{equation*}
  (\phi-u)(\phi-v)\phi_x=-1,\qquad(\phi-u)(\phi-v)\phi_y=u+v-\phi
\end{equation*}
and substituting expansion~\eqref{eq:39}, one obtains the following defining
system for the coefficients~$\phii{i}$:
\begin{align*}
  &B_{-2}\phii{-1}_x=0,
  &&B_{-2}\phii{-1}_y=0,\\
  &B_{-2}\phii{0}_x+B_{-1}\phii{-1}_x=0,
  &&B_{-2}\phii{0}_y+B_{-1}\phii{-1}_y=0,\\
  &B_{-2}\phii{1}_x+B_{-1}\phii{0}_x+B_{0}\phii{-1}_x=0,
  &&B_{-2}\phii{1}_y+B_{-1}\phii{0}_y+B_{0}\phii{-1}_y=-\phii{-1},\\
  &B_{-2}\phii{2}_x+B_{-1}\phii{1}_x+B_{0}\phii{0}_x+B_{1}\phii{-1}_x
  &&B_{-2}\phii{2}_y+B_{-1}\phii{1}_y+B_{0}\phii{0}_y+B_{1}\phii{-1}_y\\
  &\phantom{B_{-2}\phii{0}_x}=-1,
  &&\phantom{B_{-2}\phii{0}_x}=u+v-\phii{0},\\
  &B_{-2}\phii{3}_x+B_{-1}\phii{2}_x+B_{0}\phii{1}_x+B_{1}\phii{0}_x
  &&B_{-2}\phii{3}_y+B_{-1}\phii{2}_y+B_{0}\phii{1}_y+B_{1}\phii{0}_y\\ 
  &\phantom{B_{-2}\phii{0}_x}+B_{2}\phii{-1}_x=0,
  &&\phantom{B_{-2}\phii{0}_x}+B_{2}\phii{-1}_y=-\phii{1},\\
  &\phantom{B_{-2}\phii{0}_x}\dots&&\phantom{B_{-2}\phii{0}_x}\dots\\
  &B_{-2}\phii{k+2}_x+B_{-1}\phii{k+1}_x+\dots+B_{k+1}\phii{-1}_x
  &&B_{-2}\phii{k+2}_y+B_{-1}\phii{k+1}_y+\dots+B_{k+1}\phii{-1}_y\\
  &\phantom{B_{-2}\phii{0}_x}=0,&&\phantom{B_{-2}\phii{0}_x}=-\phii{k},\\
  &\phantom{B_{-2}\phii{0}_x}\dots&&\phantom{B_{-2}\phii{0}_x}\dots
\end{align*}
where
\begin{equation*}
  (\phi-u)(\phi-v)=\frac{B_{-2}}{\lambda^2} + \frac{B_{-1}}{\lambda} + B_{0} +
  B_1\lambda + \dots + B_k\lambda^k +\dots
\end{equation*}
is the expansion of the product~$(\phi-u)(\phi-v)$, i.e.,
\begin{align*}
  B_{-2}&=\left(\phii{-1}\right)^2,\\
  B_{-1}&=\phii{-1}\left(2\phii{0}-u-v\right),\\
  B_{0}&=2\phii{-1}\phii{1}+\left(\phii{0}-u\right)\left(\phii{0}-v\right),\\
  B_{1}&=2\phii{-1}\phii{2}+\left(2\phii{0}-u-v\right)\phii{1},\\
  B_{2}&=2\phii{-1}\phii{3}+\left(2\phii{0}-u-v\right)\phii{2} +
  \left(\phii{1}\right)^2,\\
  B_{3}&=2\phii{-1}\phii{4}+\left(2\phii{0}-u-v\right)\phii{3} +
  2\phii{1}\phii{2},\\ 
  &\dots\\
  B_{2k}&=2\phii{-1}\phii{2k+1}+\left(2\phii{0}-u-v\right)\phii{2k} +
  2\phii{1}\phii{2k-1}+\dots+2\phii{k-1}\phii{k+1}+\left(\phii{k}\right)^2,\\
  B_{2k+1}&=2\phii{-1}\phii{2k+2}+\left(2\phii{0}-u-v\right)\phii{2k+1} +
  2\phii{1}\phii{2k} +\dots+2\phii{k}\phii{k+1},\\
  &\dots
\end{align*}

Analysing the first four equations of the defining system, we see that the
following choice of coefficients is possible:
\begin{equation}
  \label{eq:40}
  \phii{-1}=1,\ \phii{0}=0,\ \phii{1}=-y,\ \phii{2}=-x.
\end{equation}
Then~$B_{-2}=1$, while
\begin{align*}
  B_{-1}&=-(u+v),\\
  B_{0}&=-2y+uv,\\
  B_{1}&=-2x+y(u+v),\\
  B_{2}&=2\phii{3}+x(u+v)+y^2,\\
  B_{3}&=2\phii{4}-(u+v)\phii{3}+2xy,\\
  B_{4}&=2\phii{5}-(u+v)\phii{4}-2y\phii{3}+2x^2,\\
  B_{5}&=2\phii{6}-(u+v)\phii{5}-2y\phii{4}-2x\phii{3},\\
  B_{6}&=2\phii{7}-(u+v)\phii{6}-2y\phii{5}-2x\phii{4}
  +\left(\phii{3}\right)^2,\\
  B_{7}&=2\phii{8}-(u+v)\phii{7}-2y\phii{6}-2x\phii{5}+2\phii{3}\phii{4},\\
  &\dots\\
  B_{2k}&=2\phii{2k+1}-(u+v)\phii{2k}-2y\phii{2k-1}-2x\phii{2k-2} +
  2\phii{3}\phii{2k-3}+\dots\\
  &\dots+2\phii{k-1}\phii{k+1}+\left(\phii{k}\right)^2,\\
  B_{2k+1}&=2\phii{2k+2}-(u+v)\phii{2k+1}-2y\phii{2k}-2x\phii{2k-1}+
  2\phii{3}\phii{2k-2}+\dots\\ 
  &\dots+2\phii{k}\phii{k+1},\\ 
  &\dots
\end{align*}
Hence, the initial defining system transforms to
\begin{align*}
  &\phii{3}_x=-(u+v),&&\phii{3}_y=uv-y,\\
  &\phii{4}_x=-2y-u^2-uv-v^2,&&\phii{4}_y=-x+uv(u+v),
\end{align*}
while for~$k>4$ we have the recurrent relations
\begin{equation}
\begin{array}{rcl}
  \phii{k}_x&=&B_{k-1}-B_{k-5}\phii{3}_x-\dots-B_{-1}\phii{k-1}_x,\\
  \phii{k}_y&=&B_{k-3}-B_{k-5}\phii{3}_y-\dots-B_{-1}\phii{k-1}_y-\phii{k-2}.
\end{array}\label{eq:6}
\end{equation}
Denote
by~$\Xiib{k}$ and~$\Yiib{k}$ the right-hand sides of equations~\eqref{eq:6}, 
i.e.,
\begin{equation}
  \label{XYb}
  \phii{k}_x=\Xiib{k},\quad\phii{k}_y=\Yiib{k},\qquad k\geq3.
\end{equation}
We have~$\wt{\phii{k}}=k+1$.

Now, exactly as in Subsection~\ref{sec:first-method}, we introduce the
spaces~$\bar{\mathcal{E}}_k=\bar{\mathcal{E}}_{k-1}\times\Riib{k+1}$,
$k=2,\dots$, where~$\Riib{k}=\mathbb{R}^1$ with the coordinate~$\phii{k}$,
the projections
\begin{equation*}
  \bar{\tau}_{k,k-1}\colon\bar{\mathcal{E}}_k\to\bar{\mathcal{E}}_{k-1},\qquad
  \bar{\tau}_k\colon\bar{\mathcal{E}}_k\to\mathcal{E}_1
\end{equation*}
and~$\bar{\tau}_*\colon\bar{\mathcal{E}}_*\to\mathcal{E}_1$ as the inverse
limit. We endow these spaces with the vector fields
\begin{equation*}
  \Diib{k}_x=D_x+\sum_{i=3}^{k+1}\Xiib{i}\frac{\partial}{\partial\phii{i}},\quad
  \Diib{k}_y=D_y+\sum_{i=3}^{k+1}\Yiib{i}\frac{\partial}{\partial\phii{i}}.
\end{equation*}
Similarly, we define~$\Diib{*}_x$ and~$\Diib{*}_y$.

\begin{proposition}
  For all~$k\geq2$ and~$k=*$ one has~$[\Diib{k}_x,\Diib{k}_y]=0$\textup{,}
  i.e.\textup{,} all the maps~$\bar{\tau}_{k}$ and~$\bar{\tau}_{k,k-1}$ are
  coverings. All these coverings are irreducible.
\end{proposition}

Consider the forms
\begin{equation*}
  \omegaib{k}=\Xiib{k}\diff x+\Yiib{k}\diff y.
\end{equation*}
One has~$\wt{\omegaib{k}}=k+1$ and

\begin{proposition}
  For every~$k\geq 3$\textup{,} the form~$\omegaib{k}$ is a nontrivial
  conservation law of the equation~$\bar{\mathcal{E}}_{k-1}$.
\end{proposition}

\begin{remark}
  As before, the choice~\eqref{eq:40} of initial values
  for~$\phii{-1},\dots,\phii{2}$ is not unique, but all admissible choices
  lead to equivalent results.
\end{remark}

Finally, the following statement is valid:

\begin{proposition}\label{sec:second-method-1}
  The pairs of coverings~$\tau_{k,k-1}$ and~$\bar{\tau}_{k,k-1}$, $\tau_{k}$
  and~$\bar{\tau}_{k}$, $\tau_{*}$ and~$\bar{\tau}_{*}$ are equivalent.
\end{proposition}

We provide the proof in the next subsection.

\subsection{Proof of Proposition~\ref{sec:second-method-1}}


Let us turn back to the Gibbons--Tsarev equation~\eqref{eq:21}.  For reader's
convenience, we summarise the results of the previous section in terms of the
variables $x,y,z$.  We recall that
\begin{equation*}
  \psii{0} = 0, \quad
  \psii{1} = y, \quad
  \psii{2} = x, \quad
  \psii{3} = z - \tfrac12 y^2,
\end{equation*}
while $\psii{k}$, $k > 3$, are genuine nonlocal variables of the Gibbons--Tsarev
equation, satisfying
\begin{equation*}
  \psii{k}_x = \sigma_{k-2} - \sum_{i = 1}^{k - 3} i \sigma_{k - i - 3}
  \psii{i}, \quad 
  \psii{k}_y = z_y \psii{k-1}_x - (k - 2) \psii{k - 2}.
\end{equation*}
In terms of $z$, we have
\begin{equation}
  \label{sigma by z}
  \sigma_k
  = \sum_{0 \le j \le k - j} \binom{k - j}{j} z_x^{k - 2 j} z_y^j, \quad k > 0.
\end{equation}
To prove formula~\eqref{sigma by z}, we consider the formal power series in an
auxiliary variable~$\lambda$ with coefficients taken from the two sides of
formula~\eqref{sigma by z} and show that they coincide.  Using the left-hand
side, we have, according to formula~\eqref{sigma by uv},
\begin{equation*}
  \sum_{k \ge 0} \sigma_k \lambda^k
  = \sum_{i,j \ge 0} u^i v^j \lambda^{i+j}
  = \sum_{i,j \ge 0} (u \lambda)^i (v \lambda)^j
  = \frac 1{1 - u \lambda} \cdot \frac 1{1 - v \lambda}.
\end{equation*}
Using the right-hand side, where we substitute for $z_x,z_y$ from
formulas~\eqref{eq:27}, we obtain the same series:
\begin{equation*}
  \begin{aligned}
    &\sum_{k \ge 0} \sum_{0 \le j \le k - j} \binom{k - j}{j} z_x^{k - 2 j}
    z_y^j \lambda^k 
    = \sum_{i \ge 0} \sum_{0 \le j \le i} \binom{i}{j} z_x^{i - j} z_y^j \lambda^{i + j}
    \\
    &\quad = \sum_{i \ge 0} \sum_{0 \le j \le i} \binom{i}{j} (z_x \lambda)^{i
      - j} (z_y \lambda^2)^j 
    = \sum_{i \ge 0} (z_x \lambda + z_y \lambda^2)^i
    = \sum_{i \ge 0} \bigl((u + v) \lambda - u v \lambda^2 \bigr)^i
    \\
    &\quad
    = \sum_{i \ge 0} \bigl(1 - (1 - u \lambda) (1 - v \lambda) \bigr)^i
    = \frac 1{(1 - u \lambda) (1 - v \lambda)}.
  \end{aligned}
\end{equation*}
Thus, formula~\eqref{sigma by z} is proved.

The first method of the previous section uses the expansion~\eqref{eq:32},
i.e.,
\begin{equation}
  \label{psi expansion}
  \psi(\lambda) = \lambda + \frac{\psii{1}}{\lambda}
  + \dots + \frac{\psii{k}}{\lambda^k} + \cdots,
\end{equation}
where $\psi$ satisfies the linear system~\eqref{eq:31}, which we
rewrite in terms of $z_x$, $z_y$:
\begin{equation}
  \label{psi covering}
  \psi_x = \frac{1}{\lambda^2 - z_x \lambda - z_y}\cdot\psi',\qquad
  \psi_y = \frac{\lambda - z_x}{\lambda^2 - z_x \lambda - z_y}\cdot\psi',
\end{equation}
where the `prime' denotes the $\lambda$-derivative.  The second method uses the
expansion~\eqref{eq:39}, i.e.,
\begin{equation}
  \label{phi expansion}
  \phi(\lambda) = \frac{1}{\lambda}
  + \phii{1} \lambda + \dots
  + \phii{k} \lambda^k + \cdots,
\end{equation}
where $\phi$ satisfies the nonlinear system~\eqref{eq:26}, i.e.,
\begin{equation}
  \label{phi covering}
  \phi_x = -\frac{1}{\phi^2 - z_x\phi - z_y},\qquad
  \phi_y = -\frac{\phi - z_x}{\phi^2 - z_x\phi - z_y}.
\end{equation}

Recall that composition $b \circ a$ of formal series $b(\mu) = \sum_{j \ge s} b_j \mu^j$ and 
$a(\lambda) = \sum_{i \ge r} a_i \lambda^i$, i.e.,
$$
\begin{aligned}
b(a(\lambda)) &= \sum_{j \ge s} b_j \Bigl(\sum_{i \ge r} a_i \lambda^i\Bigr)^j
\\
 &= \sum_{j \ge s} b_j \Bigl(\sum_{i_1 \ge r} a_{i_r} \lambda^{i_1}\Bigr) \dots \Bigl(\sum_{i_j \ge r} a_{i_j} \lambda^{i_j}\Bigr) 
\\
 &= \sum_{j \ge s} b_j \sum_{i_1 \ge r, \dots, i_j \ge r} a_{i_1} \dots a_{i_j} \lambda^{i_1 + \dots + i_j}
\end{aligned}
$$
is a formal series if and only if the coefficients at powers of $\lambda$ are finite sums.
This is certainly the case when $\sum_{j \ge s} b_j \mu^j$ is a polynomial or when $r \ge 1$, i.e., 
when $\sum_{i \ge r} a_i \lambda^i$ is a power series without the constant term.

Computing
$$
\frac{1}{\phi(\lambda)} = \lambda - \phii1 \lambda^{k+2} 
 + \cdots
$$ 
we see that $1/\phi(\lambda)$ is a power series without the constant term and
therefore, the composition series $\psi \circ \phi$,
i.e., 
\begin{equation*}
  \psi(\phi(\lambda)) = \phi(\lambda) + \frac{\psii{1}}{\phi(\lambda)}
  + \dots + \frac{\psii{k}}{\phi(\lambda)^k} + \cdots,
\end{equation*}
is well defined.

\begin{proposition}
  \label{prop:psiphi}
  Let $\psi(\lambda)$ and $\phi(\lambda)$ be the formal expansions~\eqref{psi
    expansion} and~\eqref{phi expansion}, respectively.  Then each pair of the
  conditions
  \begin{enumerate}
  \item equation~\eqref{psi covering}{\rm;}
  \item equation~\eqref{phi covering}{\rm;}
  \item $\psi(\phi) = c(\lambda)$, where $c(\lambda)$ is a constant
    {\rm(}possibly depending on $\lambda${\rm),}
  \end{enumerate}
  implies the remaining condition.
\end{proposition}
\begin{proof}
  Assume that~\eqref{psi covering} and~\eqref{phi covering} hold.
  Substituting $\phi$ for $\lambda$ in~\eqref{psi covering}, an easy
  computation yields
  \begin{align*}
    (\psi(\phi))_x &= \psi_x (\phi) + \psi' (\phi) \phi_x
    = 0,
    \\
    (\psi(\phi))_y &= \psi_y (\phi) + \psi' (\phi) \phi_y
    = 0
  \end{align*}
  by virtue of~\eqref{phi covering}.  Then $\psi(\phi)$ is a constant with
  respect to~$x$ and~$y$, since the covering~\eqref{phi covering} is
  differentially connected.

  Conversely, assume that $\psi(\phi) = c(\lambda)$, where $c(\lambda)$ does
  not depend on~$x$ and~$y$.  Then $\psi(\phi)_x = \psi(\phi)_y =
  0$ and
  \begin{equation*}
    \psi_x (\phi) = -\psi' (\phi) \phi_x, \quad \psi_y (\phi) = -\psi' (\phi)
    \phi_y, 
  \end{equation*}
  which yields the equivalence of Equations~\eqref{psi covering}
  and~\eqref{phi covering}.
\end{proof}

Under the substitution $\lambda \to 1/\lambda$, the expansion~\eqref{psi
  expansion} acquires the form
\begin{equation*}
  \psi\Bigl(\frac1\lambda\Bigr) = \frac1\lambda + \psii{1}\lambda
  + \dots + \psii{k}\lambda^k + \cdots,
\end{equation*}
i.e., $\phi(\lambda)$ and $\psi(1/\lambda)$ are Laurent series of the lowest
degree $-1$. Consequently, $1/\phi(\lambda)$ and $1/\psi(1/\lambda)$ are power
series without a constant term.  So, they are composable with each other.

There is a preferable choice of the constant $c(\lambda)$ in
Proposition~\ref{prop:psiphi}.

\begin{proposition}
  The expansions $\phi$\textup{,} $\psi$ can be chosen so that
  \begin{equation*}
    \psi(\phi(\lambda)) = 1/\lambda,
  \end{equation*}
  i.e.\textup{,} the power series $1/\psi(1/\lambda)$ and $1/\phi(\lambda)$ are
  compositionally inverse one to another.
\end{proposition}
\begin{proof}
  According to Proposition~\ref{prop:psiphi}, we are free to choose
  $c(\lambda) = 1/\lambda$, i.e., $\psi(\phi(\lambda)) = 1/\lambda$.
  Substituting $1/\phi(\lambda)$ for $\lambda$ in $1/\psi(1/\lambda)$, we obtain
  $1/\psi(\phi(\lambda)) = 1/c(\lambda) = \lambda$. Hence the statement.
\end{proof}

With this choice of $c(\lambda)$, the $k$-tuples of coefficients
$\psii{1},\dots,\psii{k}$ and $\phii{1},\dots,\phii{k}$ determine each other
uniquely, thereby providing the induction step in the proof of the equalities
$\mathcal{E}_{k} = \bar{\mathcal{E}}_{k}$.  It is, however, necessary to check
that the condition $c(\lambda) = 1/\lambda$ is compatible with the choices
\begin{equation}
  \label{123}
  \begin{aligned}
    \psii{1} &= y,                & \phii{1} &= -y, \\
    \psii{2} &= x,                & \phii{2} &= -x, \\
    \psii{3} &= z - \tfrac12 y^2, \quad & \phii{3} &= -z - \tfrac12 y^2. 
  \end{aligned}
\end{equation}
made in Section~\ref{sec:infin-seri-nonl}.
To this end, we compute
\begin{align*}
  \frac1\lambda &= \psi(\phi) = \phi(\lambda) + \frac{\psii{1}}{\phi(\lambda)}
  + \dots + \frac{\psii{k}}{\phi(\lambda)^k} + \cdots
  \\&
  = \frac1\lambda + (\psii{1} + \phii{1}) \lambda
  \\&\quad
  + (\psii{2} + \phii{2}) \lambda^2
  \\&\quad
  + (\psii{3} + \phii{3} - \psii{1}\phii{1}) \lambda^3
  \\&\quad
  + (\psii{4} + \phii{4} - \psii{1}\phii{2} - 2 \psii{2}\phii{1}) \lambda^4
  \\&\quad + \cdots
\end{align*}
One easily sees that the coefficients at $\lambda^i$, $i = 1$, $2$, $3$,
vanish under the above mentioned choices and we obtain the recurrent formulas
\begin{align*}
  \psii{k}
  &=-\sum_{m \ge 1} (-1)^{m} \sum_{i_1 + \cdots + i_m = k + 1}
  \frac1k \binom{k}{m}
  \phii{i_1 - 1} \cdots \phii{i_m - 1}, \\
  \phii{k}
  &= -\sum_{m \ge 1} \sum_{i_1 + \cdots + i_m = k + 1}
  \frac1k \binom{k}{m}
  \psii{i_1 - 1} \cdots \psii{i_m - 1}
\end{align*}
that provide the needed equivalence of coverings.

\section{Nonlocal symmetries}
\label{sec:algebra-nonl-symm}

It is straightforward to compute the first-degree nonlocal shadows depending on
any number of nonlocal variables.  It may seem to be insignificant whether we
use $\psii{i}$ or $\phii{i}$, but the formulas to follow turn out to be
simpler if the latter choice is made.  Thus, we give here an explicit
description of nonlocal symmetries in the covering~$\bar{\tau}^*$ and prove
that they form the Witt algebra. As the first step, we obtain the shadows.

\subsection{The hierarchy of symmetry shadows}
\label{subs:shadows}

Consider the covering~$\bar{\tau}_*$ with the nonlocal variables $\phii{i}$
and present  the total derivatives in the form
\begin{equation*}
  \tilde D_x = D_x + \sum_i \Xiib{i} \frac{\partial}{\partial\phii{i}}, \quad
  \tilde D_y = D_y + \sum_i \Yiib{i} \frac{\partial}{\partial\phii{i}},
\end{equation*}
where~$\Xiib{i}$ and $\Yiib{i}$ are the right-hand sides in~\eqref{eq:6}.

Now, using the expansion~\eqref{phi expansion}, let us introduce a new set of
nonlocal variables $\phi_{\lambda^i} = \diff^i\phi/\!\diff\lambda^i$ and
consider the product
$\bar{\mathcal{E}}_\lambda = \mathcal E \times J(\lambda;\phi)$, where
$J(\lambda;\phi)$ is the space with the coordinates~$\lambda$
and~$\phi_{\lambda^i}$, and the covering
$\bar{\tau}_\lambda\colon \bar{\mathcal{E}}_\lambda \to \mathcal{E}$. In what
follows we abbreviate the `index' $\lambda^n$ as $\Lambda$. 
We equip $\bar{\mathcal{E}}_\lambda$ with the
total derivatives
\begin{equation}\label{eq:7}
  \tilde D_x = D_x + \sum_\Lambda \phi_{x\Lambda}
  \frac{\partial}{\partial\phi_\Lambda}, \quad 
  \tilde D_y = D_y + \sum_\Lambda \phi_{y\Lambda}
  \frac{\partial}{\partial\phi_\Lambda}, \quad 
  \tilde D_\lambda = \frac{\diff}{\diff\lambda} + \sum_\Lambda
  \phi_{\lambda\Lambda} \frac{\partial}{\partial\phi_\Lambda},
\end{equation}
where the coefficients~$\phi_{x\Lambda}$ and~$\phi_{y\Lambda}$ can be computed
by means of Equations~\eqref{phi covering}. Then $\tilde{\mathcal E}$ endowed
with the vector fields~\eqref{eq:7} is equivalent to the system consisting of
the Gibbons--Tsarev equation~\eqref{eq:21}, the condition
\begin{equation}
  \label{dummy}
  z_\lambda = 0,
\end{equation}
and the pair~\eqref{phi covering} over the extended set of independent
variables $x$, $y$, $\lambda$.

\begin{proposition}
  \label{almost symmetry}
  Denote
  \begin{equation}
    \label{eq:almSymm}
    Z = (\phi^2 - z_x \phi - z_y) \phi_\lambda^2,
  \end{equation}
  Under the expansion~\eqref{phi expansion}\textup{,} $Z$ is a formal Laurent
  series of the form
  \begin{equation}
    \label{Z expansion}
    Z = \sum_{n = -4}^\infty \Z{n} \lambda^{n-2}.
  \end{equation}
  Then $\Z{n}$ are shadows of symmetries of the Gibbons--Tsarev equation
    in the covering $\bar{\tau}_*$.
\end{proposition}
\begin{proof}
  It is a routine computation to insert~\eqref{eq:almSymm}
  into the linearisation
  \begin{equation}\label{eq:tilde ell GT}
    \tilde\ell_{\mathcal{E}}(Z)\equiv \tilde D_y^2(Z) + z_x \tilde D_x \tilde
    D_y(Z) - 
    z_y \tilde D_x^2(Z) + z_{xy} \tilde D_x(Z) - z_{xx} \tilde D_y(Z)
  \end{equation}
  and check that $\tilde\ell_{\mathcal E}(Z) = 0$ modulo
  equations~\eqref{eq:21}, \eqref{phi covering} and~\eqref{dummy}.  If $Z$ is
  replaced with its expansion~\eqref{Z expansion}, we obtain
  \begin{equation*}
   0 = \tilde\ell_{\mathcal E}(Z)
    = \sum_{n = -4}^\infty \tilde\ell_{\mathcal E}(\Z{n}) \lambda^{n-2}.
  \end{equation*}
  Since $\tilde\ell_{\mathcal E}(\Z{n})$ do not depend on $\lambda$, they 
  have to vanish 
  modulo equation~\eqref{eq:21} and expanded system~\eqref{phi covering}, i.e., 
  equations \eqref{XYb}.  
  Hence the statement.
\end{proof}

\begin{remark} \label{rem:almost symmetry a}
It is easy to compute functions $Z$ such that $\tilde\ell_{\mathcal E}(Z) = 0$ 
modulo equations~\eqref{eq:21}, \eqref{phi covering} and~\eqref{dummy}
(cf. the proof of Proposition~\ref{almost symmetry}).
Besides the expression~\eqref{eq:almSymm}, another such function is 
$Z = \phi_\lambda$, which, however, generates just the invisible symmetries
(see Sect.~\ref{sec:prel-notat}). 

Moreover, if some $Z$ satisfies $\tilde\ell_{\mathcal E}(Z) = 0$, then so does 
$f(\lambda) Z$ for any function $f(\lambda)$. This does not 
extend the linear space of generated shadows $\Z{i}$, however.
\end{remark}

\begin{remark} \label{rem:almost symmetry b}
Although the condition $z_\lambda = 0$ is necessary for $Z$ given 
by~\eqref{eq:almSymm} to be a shadow of the Gibbons--Tsarev equation, 
the same $Z$ does not satisfy $\tilde D_\lambda Z = 0$ and, therefore, is not a 
shadow of the system consisting of the Gibbons--Tsarev equation and the 
equation $z_\lambda = 0$.
\end{remark}

Proposition~\ref{almost symmetry} says that $Z$ is the generating section for
an infinite hierarchy of shadows of the Gibbons--Tsarev equation.  These
shadows are easy to obtain explicitly.  Let $\rsum\quad$ denote summation
where indices run through all integers from $-1$ to infinity, possibly
subject to additional requirements written under the symbol.


\begin{proposition}
  \label{prop:shadow}
  Let
  \begin{equation}
    \begin{aligned}
      \label{A2}
      A_2^{(k,n)} &= \rsum_{i_1 + \cdots + i_{k+2} = n} i_{1} i_{2} \phii{i_1}
      \cdots \phii{i_{k+2}}, 
      \quad k \ge 0.
    \end{aligned}
  \end{equation}
  Then
  \begin{equation*}
    \Z{n} = A_2^{(1,n)} z_x + A_2^{(0,n)} z_y - A_2^{(2,n)}.
  \end{equation*}
\end{proposition}

\begin{proof}
  Considering the expansion~\eqref{phi expansion}, we have
  \begin{equation*}
    \begin{aligned}
      \phi^k \phi^{\prime2} &= \Bigl(\rsum_{i_1}\quad \phii{i_1} \lambda^{i_1}
      \Bigr) \cdots \Bigl(\rsum_{i_k}\quad \phii{i_k} \lambda^{i_k} \Bigr)
      \\&\qquad\times \frac{1}{\lambda^2} \Bigl(\rsum_{i_{k+1}}\quad i_{k+1}
      \phii{i_{k+1}} \lambda^{i_{k+1}} \Bigr) \Bigl(\rsum_{i_{k+2}}\quad
      i_{k+2} \phii{i_{k+2}} \lambda^{i_{k+2}} \Bigr)
      \\
      &= \sum_n \rsum_{i_1 + \cdots + i_{k+2} = n} i_{k+1} i_{k+2} \phii{i_1}
      \cdots \phii{i_{k+2}} \lambda^{n-2} = \sum_n A_2^{(k,n)} \lambda^{n-2}.
    \end{aligned}
  \end{equation*}
  Inserting into $Z$ given by formula~\eqref{eq:almSymm}, we obtain the result
  immediately.
\end{proof}


\subsection{The hierarchy of full symmetries}
\label{subs:symmetries}

Here the shadows $\Z{n}$ obtained in the previous section will be extended to
full symmetries of the covering $\bar{\tau}_*$. To this end, consider a
nonlocal symmetry in the form 
\begin{equation*}
  a \frac\partial{\partial x}
  + b \frac\partial{\partial y}
  + c \frac\partial{\partial z}
  + \sum_{i > 3} f^{(i)} \frac{\partial}{\partial\phii{i}}
  + \cdots,
\end{equation*}
where $a$, $b$, $c$, $f^{(i)}$ are functions on $\bar{\mathcal E}_*$. Then the
corresponding vertical field, obtained by subtracting
$a \tilde D_x + b \tilde D_y$, is
\begin{equation}
  \label{Symm xyz phi>3}
  \begin{aligned}
    \mathcal S
    &= \sum_{\Xi} \tilde D_\Xi (c - a z_x - b z_y) \frac\partial{\partial z_\Xi}
    + \sum_{i > 3} (f^{(i)} - a \phii{i}_x - b \phii{i}_y)
    \frac{\partial}{\partial\phii{i}},
  \end{aligned}
\end{equation}
where~$\phii{i}_x = \Xiib{i}$, $\phii{i}_y = \Yiib{i}$ are given by recurrent
relations~\eqref{eq:6} and~$z_\Xi$ are internal coordinates in~$\mathcal{E}$
(see Section~\ref{sec:local-symm-cons}). Then $\mathcal S$ is a
symmetry of $\bar{\mathcal E}_*$ if and only if
\begin{equation}
  \label{S cond phi(i)}
  \begin{aligned}
    &\mathcal S(z_{yy} + z_x z_{xy} - z_y z_{xx} + 1)= 0,
    \quad
    &\mathcal S(\phii{i}_x - \Xiib{i})= 0,
    \quad
    &\mathcal S(\phii{i}_y - \Yiib{i}) = 0
  \end{aligned}
\end{equation}
modulo equations~\eqref{eq:21} and~\eqref{XYb}.
Using formulas~\eqref{123}, variables $x,y,z$ can be expressed in terms of
$\phii{1}, \phii{2}, \phii{3}$.  Consequently, we can rewrite~\eqref{Symm xyz
  phi>3} and~\eqref {S cond phi(i)} in terms of $\phii{i}$ and $z_\Xi$,
$|\Xi| > 0$, alone.

\begin{proposition}
  In terms of coordinates $\phii{i}$\textup{,} $i > 0$\textup{,} and
  $z_\Xi$\textup{,} a vertical evolutionary field in the covering
  $\bar{\tau}_*$ can be written as
  \begin{equation}
    \label{Symm phii}
    \begin{aligned}
      \mathcal S &=
      \sum_{i > 0} \Phi^{(i)}
      \frac{\partial}{\partial\phii{i}}
      + \sum_{|\Xi| > 0} \tilde D_\Xi Z \frac\partial{\partial z_\Xi},
      \\
      Z &= (z_y - \phii1) f^{(1)} + z_x f^{(2)} - f^{(3)},
      \\
      \Phi^{(i)} &= f^{(i)} + f^{(2)} \Xiib{i} + f^{(1)} \Yiib{i}.
    \end{aligned}
  \end{equation}
  The field $\mathcal S$ is a symmetry if and only if $\tilde\ell_{\mathcal E}
  Z = 0$ and
  \begin{equation}
    \label{nloc sym phi (i)}
      \tilde D_x \Phi^{(i)}
      - \mathcal S \Xiib{i}= 0, \qquad
      \tilde D_y \Phi^{(i)}
      - \mathcal S \Yiib{i} = 0.
  \end{equation}
\end{proposition}

\begin{proof}
  Formulas~\eqref{Symm phii} are obtained by direct computation, while
  \eqref{nloc sym phi (i)} follows from~\eqref{S cond phi(i)}
  immediately.
\end{proof}

Let us now pass from the covering~$\bar{\tau}_*$ with the nonlocal variables
$\phii{i}$ to the covering $\bar{\mathcal E}_\lambda \to \mathcal E$ obtained
from the covering~\eqref{phi covering} by means of the expansion~\eqref{phi
  expansion}, i.e.,
\begin{equation*}
  \phi = \frac{1}{\lambda}
  + \sum_i \phii{i} \lambda^i.
\end{equation*}
Then $\phi_\lambda = -1/\lambda^2 + \sum_i i \phii{i} \lambda^{i-1}$,
$\phi_{\lambda\lambda} = 2/\lambda^3 + \sum_i i (i - 1) \phii{i}
\lambda^{i-2}$, etc.  Hence,
\begin{equation*}
  \frac{\partial\phi}{\partial\phii{i}} = \lambda^i,
  \quad
  \frac{\partial\phi_\lambda}{\partial\phii{i}} = i \lambda^{i-1} =
  \frac{\diff\lambda^i}{\diff\lambda}, 
  \quad
  \frac{\partial\phi_{\lambda\lambda}}{\partial\phii{i}} = i (i - 1) \lambda^{i-2}
  = \frac{\diff^2\lambda^i}{\diff\lambda^2}, \dots
\end{equation*}
and, therefore,
\begin{equation*}
  \frac{\partial}{\partial\phii{i}} = \sum_\Lambda \partial_\Lambda \lambda^i
  \frac\partial{\partial\phi_\Lambda}, 
\end{equation*}
where, as above, $\Lambda$ stands for $\lambda^n$, $n\geq0$, and
$\partial_{\lambda^n} = {d^n}/{d\lambda^n}$.  Alternatively speaking, the
vector field $\partial/\partial\phii{i}$, when rewritten in the coordinates
$\lambda,\phi_\Lambda$, is the prolongation of
$\lambda^i\,\partial/\partial\phi$.  For completeness, we note that the vector
field written as $\partial/\partial\lambda$ in the coordinates
$\lambda,\phii1,\phii2,\phii3, \dots$ becomes
\begin{equation*}
  \frac{\partial}{\partial\lambda}
  + \phi_\lambda \frac{\partial}{\partial\phi}
  + \phi_{\lambda\lambda} \frac{\partial}{\partial\phi_\lambda}
  + \dots
  = \frac{\partial}{\partial\lambda}
  + \sum_\Lambda \phi_{\lambda\Lambda} \frac{\partial}{\partial\phi_\Lambda} =
  D_\lambda 
\end{equation*}
in the coordinates $\lambda,\phi_\Lambda$.

\begin{proposition} \label{prop:Symm phi_Lambda} In terms of the coordinates
  $\phi_\Lambda$ and $z_\Xi$\textup{,} $|\Xi| > 0$\textup{,} a vertical
  infinitely prolonged field in the covering $\bar{\tau}_\lambda$ can be
  written as
  \begin{equation}
    \label{Symm phi_Lambda}
    \mathcal S =
    \sum_{\Lambda} \partial_\Lambda \Phi
    \frac{\partial}{\partial\phi_\Lambda}
    + \sum_{|\Xi| > 0} \tilde D_\Xi Z \frac\partial{\partial z_\Xi},
  \end{equation}
  where
  \begin{equation}
    \label{Z Phi f}
    \begin{aligned}
      &Z = (z_y - \phii{1}) f^{(1)} + z_x f^{(2)} -  f^{(3)},\\
      &\Phi = f - \frac{f^{(2)} + f^{(1)}\,(\phi - z_x)}{\phi^2 - z_x \phi -
        z_y},\\ 
      &f = \sum_{i > 0} f^{(i)}\lambda^i.
    \end{aligned}
  \end{equation}
  The field $\mathcal S$ is a symmetry if and only if
  $\tilde\ell_{\mathcal E}(Z) = 0$\textup{,} see~\eqref{eq:tilde ell
    GT}\textup{,} and
  \begin{equation}
    \label{nloc sym phi}
    \begin{aligned}
      &\tilde D_x \Phi
      + \frac{\phi \tilde D_x Z + \tilde D_y Z - (2\phi - z_x) \Phi}{(\phi^2 -
        z_x \phi - z_y)^2} 
      = 0, \\
      &\tilde D_y \Phi
      + \frac{z_y \tilde D_x Z + (\phi - z_x) \tilde D_y Z - ((\phi - z_x)^2 +
        z_y) \Phi} 
      {(\phi^2 - z_x \phi - z_y)^2}
      = 0.
    \end{aligned}
  \end{equation}
\end{proposition}

\begin{proof}
  Formula~\eqref{Symm phii} can be rewritten as
  \begin{equation*}
    \begin{aligned}
      \mathcal S &=
      \sum_{i > 0} \Phi^{(i)}
      \frac{\partial}{\partial\phii{i}}
      + \sum_{|\Xi| > 0} \tilde D_\Xi Z \frac\partial{\partial z_\Xi}
      \\
      &=
      \sum_\Lambda
      \partial_{\Lambda} \Bigl(\sum_{i > 0}\Phi^{(i)}\lambda^i\Bigr)
      \frac\partial{\partial\phi_\Lambda}
      + \sum_{|\Xi| > 0} \tilde D_\Xi Z \frac\partial{\partial z_\Xi},
    \end{aligned}
  \end{equation*}
  where $Z = (z_y - \phii1) f^{(1)} + z_x f^{(2)} - f^{(3)}$, and
  \begin{equation*}
    \begin{aligned}
      \sum_{i > 0}\Phi^{(i)}\lambda^i &= \sum_{i > 0} f^{(i)}\lambda^i
      + f^{(2)} \sum_{i > 0} \Xiib{i}\lambda^i + f^{(1)} \sum_{i > 0}
      \Yiib{i}\lambda^i 
      \\
      &= f
      + f^{(2)} \sum_{i > 0} \phii{i}_x\lambda^i + f^{(1)} \sum_{i > 0}
      \phii{i}_y\lambda^i\\ 
      &= f + f^{(2)} \phi_x + f^{(1)} \phi_y
      = f - \frac{f^{(2)} + f^{(1)}\,(\phi - z_x)}{\phi^2 - z_x \phi - z_y},
    \end{aligned}
  \end{equation*}
  where $f = \sum_{i > 0} f^{(i)}\lambda^i$.  Hence formulas~\eqref{Z Phi f}.

  Now, $\mathcal S$ is a symmetry if and only if
  \begin{equation*}
    \tilde D_x \Phi
    + \mathcal S \left(\frac{1}{\phi^2 - z_x \phi - z_y}\right) = 0, \qquad
    \tilde D_y \Phi
    + \mathcal S \left(\frac{\phi - z_x}{\phi^2 - z_x \phi - z_y}\right) = 0.
  \end{equation*}
  These are formulas~\eqref{nloc sym phi}.
\end{proof}

The last proposition suggests the following construction. Let
  \begin{equation*}
    f = \sum_{i > 0} f^{(i)} \lambda^i,
  \end{equation*}
where the coefficients~$ f^{(i)}$ are independent of~$\lambda$. Then we set
  \begin{equation*}
    \PBS f = \sum_{i > 0} f^{(i)} \frac{\partial}{\partial\phii{i}}.
  \end{equation*}
Transforming to the coordinates $\lambda,\phi,\dots,\phi_\Lambda,\dots$, we
obtain
\begin{align*}
  \PBS f = \sum_{\Lambda} \partial_{\Lambda} f
  \frac{\partial}{\partial\phi_\Lambda}
\end{align*}
which is the usual prolongation of a vertical generator
$f\,\partial/\partial\phi$.  Obviously,
\begin{equation}
  [\PBS f, \PBS g] = \PBS{\{f,g\}},
  \quad
  \{f,g\} = \PBS f g - \PBS g f.
\end{equation}
Using this notation, the symmetries we are looking for can be written as
\begin{equation*}
  \mathcal S =
  \PBS\Phi
  + \sum_{|\Xi| > 0} \tilde D_\Xi Z \frac\partial{\partial z_\Xi},
\end{equation*}
where
\begin{equation*}
  \begin{aligned}
    &Z = (z_y - \phii{1}) f^{(1)} + z_x f^{(2)} -  f^{(3)},\\
    &\Phi = f - \frac{f^{(2)} + f^{(1)}\,(\phi - z_x)}{\phi^2 - z_x \phi - z_y},\\
    &f = \sum_{i > 0} f^{(i)}\lambda^i,
  \end{aligned}
\end{equation*}
and should satisfy $\tilde\ell_{\mathcal E}(Z) = 0$, see~\eqref{eq:tilde ell
  GT}, as well as Equations~\eqref{nloc sym phi}.

Now, using the formal series
\begin{equation}
  \label{triv symm}
  f = \lambda^n \phi_\lambda
  = -\lambda^{n-2} + \sum_{i \ge 1} i \phii{i} \lambda^{n+i-1}
\end{equation}
and the field $\PBS{f}$ we shall show that all the shadows described in
Subsection~\ref{subs:shadows} are lifted to a nonlocal symmetry
in~$\bar{\tau}_*$. The proof depends on the integer~$n$.

\textbf{The case $n\ge 3$.} The series $\lambda^n \phi_\lambda$ is of the form
required by the definition of $\PBS f$.  In all these cases, conditions
\eqref{eq:tilde ell GT} and~\eqref{nloc sym phi} are easily checked by
straightforward computation, which is omitted.

For $n = 3$, we have
$ f = \lambda^3 \phi_\lambda = -\lambda + \sum_{i \ge 1} i \phii{i}
\lambda^{2+i}$,
i.e., $f^{(1)} = -1$, $f^{(2)} = 0$, $f^{(3)} = \phii{1}$.  In this case,
$Z = -z_y = -\Zi{-2}$, i.e., we obtain the lift
\begin{equation*}
  \PBS{\lambda^{3}\phi_\lambda} = -\frac{\partial}{\partial\phii{1}}
  + \sum_{i \ge 1} i \phii{i} \frac{\partial}{\partial\phii{2+i}}
\end{equation*}
of the $y$-translation.

For $n = 4$, we have
$ f = \lambda^4 \phi_\lambda = -\lambda^{2} + \sum_{i \ge 1} i \phii{i}
\lambda^{3+i}$,
i.e., $f^{(1)} = f^{(3)} = 0$, $f^{(2)} = -1$.  In this case,
$Z = -z_x = -\Zi{-3}$, i.e., we obtain the lift
\begin{equation*}
  \PBS{\lambda^{4}\phi_\lambda} = -\frac{\partial}{\partial\phii{2}}
  + \sum_{i \ge 1} i \phii{i} \frac{\partial}{\partial\phii{3+i}}
\end{equation*}
of the $x$-translation.

If $n = 5$, then
$ f = \lambda^5 \phi_\lambda = -\lambda^{3} + \sum_{i \ge 1} i \phii{i}
\lambda^{4+i}$,
i.e., $f^{(1)} = f^{(2)} = 0$, $f^{(3)} = -1$.  Obviously, $Z = 1 = \Zi{-4}$,
i.e., we recover the first classical symmetry and obtained its lift
\begin{equation*}
  \Si{-4} = \PBS{\lambda^{5}\phi_\lambda} = -\frac{\partial}{\partial\phii{3}}
  + \sum_{i \ge 1} i \phii{i} \frac{\partial}{\partial\phii{4+i}}.
\end{equation*}

If $n \ge 6$, then the coefficients $f^{(1)}$, $f^{(2)}$, $f^{(3)}$ are
zero. Obviously, $Z = 0$ and we obtain the invisible symmetries
\begin{equation}
  \label{S1-n}
  \Si{1-n} =\PBS{\lambda^{n}\phi_\lambda} =
  -\frac{\partial}{\partial\phii{n-2}} 
  + \sum_{i \ge 1} i \phii{i} \frac{\partial}{\partial\phii{n+i-1}},
  \quad n \ge 6.
\end{equation}

\textbf{The case $n< 3$.} In this case, the series~\eqref{triv symm} contains
non-positive terms and so we cannot construct the corresponding
field~$\PBS{f}$ directly.  To overcome this problem, we do the following.

For any formal series $T=\sum_{i} c_i t^i$ we use the notation~$[t^n] T =
c_n$ and define the operators
\begin{equation*}
   P_\varphi f
   = \sum_{k = 0}^m [\varphi^{k}] f \cdot \varphi^{k},\qquad
   \mathfrak P_\varphi f = f - P_\varphi f =f-\sum_{k=0}^{m}[\varphi^{k}] f
  \cdot \varphi^{k}.
\end{equation*}
Then $\mathfrak P_\varphi f$ is a positive series in~$\lambda$ and a polynomial
in~$\phi$.

Consider now the family
$f_{n} = \mathfrak P_\varphi(\lambda^n \varphi_\lambda)$, $n\leq2$.  The first
members are
$ f_{2}= \mathfrak P_\varphi(\lambda^2 \varphi_\lambda) = \lambda^2
\varphi_\lambda + 1, f_{1}=\mathfrak P_\varphi(\lambda \varphi_\lambda) =
\lambda \varphi_\lambda + \varphi, f_{0}=\mathfrak P_\varphi(\varphi_\lambda)
= \varphi_\lambda + \varphi^2 - 3\phii{1}$, etc.

\begin{proposition}\label{sec:hier-full-symm-1}
  For any $n\leq2$\textup{,} all the vector fields $\mathcal{S}$ of the form
  \eqref{Symm phi_Lambda} with
  \begin{equation}\label{fngen}
    f = f_{n} = \mathfrak P_\varphi\zav{\varphi_\lambda
      \lambda^{n}}
  \end{equation}
  are nonlocal symmetries in~$\bar{\tau}_*$.
\end{proposition}

\begin{example}
  To illustrate how the construction works, let us discuss the case of $n = 1$
  in more detail.  We have
  \begin{equation*}
    \lambda \phi_\lambda = -\frac1\lambda
    + \phii{1} \lambda + 2 \phii{2} \lambda^2 + 3 \phii{3} \lambda^3 + \dots
  \end{equation*}
  and
  \begin{equation*}
    f_1 = \mathfrak{P}_{\lambda \phi_\lambda} = \lambda \phi_\lambda + \phi
    = 2 \phii{1} \lambda + 3 \phii{2} \lambda^2 + 4 \phii{3} \lambda^3 + \dots  
  \end{equation*}
  Consequently, $f_1^{(i)} = (i + 1) \phii{i}$. In particular,
  $f^{(1)} = 2 \phii{1} = -2y$, $f^{(2)} = 3 \phii{2} = -3x$,
  $f^{(3)} = 4 \phii{3} = -4z - 2 y^2$.  Substituting into formulas~\eqref{Z
    Phi f}, we get
  \begin{equation*}
    \begin{aligned}
      Z & = (z_y - \phii{1}) f^{(1)} + z_x f^{(2)} -  f^{(3)}
      = -2y z_y - 3x z_x + 4z
      = -Z^{(0)},\\
      \Phi &= \lambda \phi_\lambda + \phi + \frac{3x + 2y\,(\phi -
        z_x)}{\phi^2 - z_x \phi - z_y}. 
    \end{aligned}
  \end{equation*}
  Thus, we have the lift of the scaling symmetry.
\end{example}

To prove Proposition~\ref{sec:hier-full-symm-1}, we introduce the
\textit{generating function}
\begin{equation}\label{genfunctionbl}
  f= f(\lambda,\xi)=\sum_{n=-4}^{\infty}\xi^{n-2} f_{1-n}.
\end{equation}
If we show that $f$ satisfies equations \eqref{nloc sym phi}, then,
by linearity, \eqref{nloc sym phi} will be satisfied for all $f_n$. 
To turn this observation into a proof, we need an analytic expression 
for $f$ and also for the first three coefficients $f^{(i)}$ of the expansion
$f = \sum_i f^{(i)}(\xi) \lambda^i$.

\begin{lemma}\label{genfunctionlemma} The generating function
  \eqref{genfunctionbl} admits the representation
  \begin{equation*}
    f=\zav{\frac{\lambda}{\xi}}^{6}\frac{\varphi_\lambda(\lambda)}{\lambda-\xi}
    +\frac{\varphi_\xi(\xi)^2}{\varphi(\xi)-\varphi(\lambda)}. 
  \end{equation*}
  In addition\textup{,} 
  \begin{equation*}
    f^{(1)}=-\varphi_\xi(\xi)^2,\qquad
    f^{(2)}=-\varphi_\xi(\xi)^2\varphi(\xi),\qquad 
    f^{(3)}=\varphi_\xi(\xi)^2\zav{\varphi^{(1)}-\varphi(\xi)^2}.
  \end{equation*}
\end{lemma}

\begin{proof}
  We break the proof into several steps.

  \medskip
  \noindent \textit{Step 1.} We show that
  \begin{equation*}
    \hzav{\varphi(\xi)^k}\frac{\varphi_\xi(\xi)}{\xi^{n-1}}=
    -\hzav{\xi^{n-2}}\frac{\varphi_\xi(\xi)^2}{\varphi(\xi)^{k+1}}. 
  \end{equation*}
  Recall that the `formal residue' of the Laurent series
  $g(\xi)=\sum_{k=-\infty}^{\infty} g_k \xi^k$ is defined by
  \begin{equation*}
    \Res{g}=\hzav{\xi^{-1}}g(\xi)=g_{-1}.
  \end{equation*}
  It is straightforward to check that it has the following properties:
  \begin{eqnarray*}
    \Res{\alpha g+\beta h}&=&\alpha \Res{g}+\beta\Res{h},\qquad \alpha,\beta
    =\const, \\ 
    \Res{g^\prime}&=& 0,\\
    \hzav{\xi^n}g &=&\Res{\frac{g}{\xi^{n+1}}},\\
    \Res{(g(h)) h'}&=&\Res{\frac{h^\prime}{h}}\cdot\Res{g},
  \end{eqnarray*}
  where the `prime' denotes the $\xi$-derivative and $g(h)$ is the
  composition of formal series.  The last property is valid for all Laurent series $g$ of the
  form $
    g=g_{-n}\xi^{-n}+g_{1-n}\xi^{1-n}+\cdots + g_0+g_1 \xi+g_2 \xi^2+\cdots$,
  i.e., whose principal part is finite, and for $h$ of the form $h=
  \xi^m(h_0+h_1 \xi+h_2 \xi^2+\cdots )$ or $h=\xi^{-m}(h_0+h_1 \xi^{-1}+h_2
  \xi^{-2}+\cdots)$, where $m>0$.

  We are going to use the last property in the form
  \begin{equation*}
    \Res{g}=\frac{\Res{(g\circ h) h'}}{\Res{({h'}/{h})}},
  \end{equation*}
  with $g(\omega)=\psi(\omega)$ and $h(\xi)=\varphi(\xi)$, where $\psi$ is the
  compositional inverse of $\varphi$,
  i.e. $\psi(\varphi(\lambda))=\varphi(\psi(\lambda))=\lambda$.  In other
  words $\omega=\varphi(\xi)$, $\xi=\psi(\omega)$. Since~$\phi$ and~$\psi$ are
  compositionally mutually inverse, one has
  \begin{equation*}
    \varphi'(\psi(\omega))=\frac{1}{\psi'(\omega)},\qquad
    \psi'(\varphi(\xi))=\frac{1}{\varphi'(\xi)},
  \end{equation*}
  and using the obvious identity $\Res{({\varphi'(\xi)}/{\varphi(\xi)})}=-1$,
  we obtain
  \begin{align*}
    \hzav{\varphi(\xi)^k}\frac{\varphi'(\xi)}{\xi^{n-1}}& =
    \hzav{\omega^k}\frac{1}{\psi'(\omega)\psi(\omega)^{n-1}} =
    \Res{\frac{1}{\psi'(\omega)\psi(\omega)^{n-1}\omega^{k+1}}}\\ 
    &=-\Res{\frac{\varphi'(\xi)}{\psi'(\varphi(\xi))\xi^{n-1}\varphi(\xi)^{k+1}}
    }=
    -\Res{\frac{\varphi'(\xi)^2}{\xi^{n-1}\varphi(\xi)^{k+1}}}=
    -\hzav{\xi^{n-2}}\frac{\varphi'(\xi)^2}{\varphi(\xi)^{k+1}}.
  \end{align*}
  \medskip

  \noindent\textit{Step 2.} Substituting this result into the definition of
  $f_n$ we obtain
  \begin{equation}\label{fncoeff}
    f_{1-n}=\frac{\varphi'(\lambda)}{\lambda^{n-1}} +
    \sum_{k=0}^{n+1}\varphi(\lambda)^k\hzav{\xi^{n-2}}
    \frac{\varphi'(\xi)^2}{\varphi(\xi)^{k+1}},  
  \end{equation}
  where, as a matter of fact, we can extend the upper summation bound to
  infinity since
  \begin{equation*}
    \hzav{\xi^{n-2}}\frac{\varphi'(\xi)^2}{\varphi(\xi)^{k+1}}=0,\qquad k>n+1.
  \end{equation*}
  Thus, we obtain
  \begin{equation}\label{neargenfun}
    f_{1-n}=\frac{\varphi'(\lambda)}{\lambda^{n-1}}+\hzav{\xi^{n-2}}
    \sum_{k=0}^{\infty}\varphi(\lambda)^k
    \frac{\varphi'(\xi)^2}{\varphi(\xi)^{k+1}}=   
    \frac{\varphi'(\lambda)}{\lambda^{n-1}}+
    \hzav{\xi^{n-2}}\frac{\varphi'(\xi)^2}{\varphi(\xi)-\varphi(\lambda)}. 
  \end{equation}
  \medskip

  \noindent\textit{Step 3.} To get the closed formula sought  for the
  generating  function
  \begin{equation*}
    f(\lambda,\xi)=\sum_{n=-4}^{\infty} \xi^{n-2}f_{1-n},
  \end{equation*}
  it remains to use the identities
  \begin{equation*}
    \sum_{n=-4}^{\infty} \xi^{n-2}\frac{\varphi'(\lambda)}{\lambda^{n-1}}=
    \zav{\frac{\lambda}{\xi}}^6\frac{\varphi'(\lambda)}{\lambda-\xi},\qquad
    \sum_{n=-4}^{\infty}
    \xi^{n-2}\hzav{\xi^{n-2}}
    \frac{\varphi'(\xi)^2}{\varphi(\xi)
      -\varphi(\lambda)}=\frac{\varphi'(\xi)^2}{\varphi(\xi)-
      \varphi(\lambda)}. 
  \end{equation*}
  \medskip

  \noindent\textit{Step 4.} Computation of the first three coefficients
  \begin{equation*}
    f^{(1)}=\hzav{\lambda} f(\lambda,\xi),\qquad f^{(2)}=\hzav{\lambda^2}
    f(\lambda,\xi), \qquad f^{(3)}=\hzav{\lambda^3} f(\lambda,\xi)
  \end{equation*}
  is straightforward.
\end{proof}

\begin{proof}[Proof of Proposition~\ref{sec:hier-full-symm-1}]
  Let us verify
  equations \eqref{nloc sym phi} for the generating function
  $f=f(\lambda,\xi)$.

  In this case, we have
  \begin{equation*}
    Z=-\zav{z_y+z_x\varphi(\xi)-\varphi(\xi)^2}\varphi_\xi(\xi)^2. 
  \end{equation*}
  This is a shadow of symmetries of the Gibbons--Tsarev equation as proved in
  Proposition~\ref{almost symmetry}, so Equation~\eqref{eq:tilde ell GT} is
  satisfied. We have
  \begin{equation*}
    \Phi=
    \zav{\frac{\lambda}{\xi}}^{6}\frac{\varphi_\lambda(\lambda)}{\lambda-\xi}+ 
    \frac{\varphi_\xi(\xi)^2}{\varphi(\xi)-\varphi(\lambda)}
    \frac{\varphi(\xi)^2-z_x\varphi(\xi)-z_y}{\varphi(\lambda)^2-
      z_x\varphi(\lambda)-z_y}.  
  \end{equation*}

  These expressions can be put directly into equations \eqref{nloc sym phi}.
  The proof that equations \eqref{nloc sym phi} indeed hold is a matter of
  direct computation with the help of the identities
  \begin{align*}
    \frac{\partial\varphi(t)}{\partial x}&=
    -\frac{1}{\varphi(t)^2-z_x\varphi(t)-z_y},\\ 
    \frac{\partial\varphi(t)}{\partial y}&=
    -\frac{\varphi(t)-z_x}{\varphi(t)^2-z_x\varphi(t)-z_y},\\ 
    \frac{\partial\varphi_t(t)}{\partial x}&=
    \frac{(2\varphi(t)-z_x)\varphi_t(t)}{(\varphi(t)^2-z_x\varphi(t)-z_y)^2},\\ 
    \frac{\partial\varphi_t(t)}{\partial y}&=
    -\frac{((\varphi(t)-z_x)^2+z_y)
      \varphi_t(t)}{(\varphi(t)^2-z_x\varphi(t)-z_y)^2},\\  
    z_{yy}&=z_yz_{xx}-z_xz_{xy}-1,
  \end{align*}
  where $t$ is either $\lambda$ or $\xi$.
\end{proof}

Collecting together the above facts, we obtain the main result of this
section:

\begin{theorem} In the notation of Proposition \ref{prop:Symm
    phi_Lambda}\textup{,} the vector fields $\mathcal{S}$ with $f=f_n$
  \eqref{fngen} are symmetries for all $n\in \mathbb{Z}$.
\end{theorem}

Remarkably, we are able to obtain explicit formulas for the symmetries.
Denote by
\begin{equation*}
  A_r^{(k,n)} = [\lambda^{n-r}] \phi^k {\phi'}^{r},
\end{equation*}
the coefficient at $\lambda^{n-r}$ in the product $\phi^k {\phi'}^r$, where
$k,r$ are arbitrary integers, cf.~\eqref{A2}.
For $r=0$ we have
\begin{equation*}
  A_0^{(k,n)} = \rsum_{j_1 + \cdots + j_k = n} \phii{j_1} \phii{j_2} \cdots
  \phii{j_k}, \quad k > 0
\end{equation*}
(recall that the notation $\rsum_{\dots}{}$ means
  summation where indices run through all integers from $-1$ to infinity).

\begin{proposition} 
  Vector fields $\Si{n} = \PBS{f_{1-n}}$ admit the explicit formula
  \begin{equation}
    \Si{n}= \sum_{m \ge 1} \biggl((n+m) \phii{n+m} + \sum_{k = 0}^{n+1}
    A_0^{(k,m)} A_2^{(-k-1,n)} \biggr) 
    \frac{\partial}{\partial \phii{m}}.
  \end{equation}
\end{proposition}
\begin{proof}
  This is a direct consequence of the representation of $f_{1-n}$,
  see~\eqref{fncoeff}, and the definition of $\PBS{f}$.
\end{proof}
\begin{remark}
  Alternatively, we can also write
  \begin{equation*}
    \Si{n} = - \sum_{m \ge 1} \sum_{k = 0}^{m} A_0^{(-k-1,m)} A_2^{(k,n)} 
    \frac{\partial}{\partial \phii{m}}.
  \end{equation*}
\end{remark}

\subsection{The Lie algebra structure}
\label{sec:lie-algebra-struct}
The main result of this part is
\begin{theorem}\label{sec:hier-full-symm}
  The vector fields
  \begin{equation*}
    \Si{n} = \PBS{f_{1-n}}
  \end{equation*}
  satisfy
  \begin{equation*}
    [\Si{n},\Si{m}] = (m-n) \Si{n+m},
  \end{equation*}
  i.e.\textup{,} constitute a basis of the Witt algebra.
\end{theorem}

For the proof we are going to make use of the following lemma:
\begin{lemma}\label{comrule}
  Let
  \begin{equation*}
    g=g(\varphi,\varphi^\prime,\lambda),\quad
    h=h(\varphi,\varphi^\prime,\lambda),
  \end{equation*}
  be two formal series in~$\lambda$ of the lowest order~$1$. Then
  \begin{equation}
    [\PBS{g},\PBS{h}]=\PBS{\szav{g,h}},
  \end{equation}
  where
  \begin{equation*}
    \szav{g,h}=h_\varphi g- g_\varphi h+h_{\varphi^\prime} D_\lambda
    g-g_{\varphi^\prime} D_\lambda h,
  \end{equation*}
  where\textit{,} as before\textit{,} $D_\lambda = \partial/\partial \lambda +
  \sum_\Lambda \phi_{\lambda\Lambda}\partial/\partial\Lambda$.
\end{lemma}
\begin{proof}
  Obviously from~\eqref{phi expansion}, one has
  \begin{equation*}
    \frac{\partial\phi}{\partial\phii{m}} = \lambda^m, \quad
    \frac{\partial\phi'}{\partial\phii{m}} = m \lambda^{m-1}
  \end{equation*}
  for all integers $m > 0$.  Therefore,
  $ \partial_{\varphi_m} g(\lambda)= g_\varphi \lambda^m+g_{\varphi^\prime}
  D_\lambda \lambda^m$.  Hence
  \begin{equation*}
    \PBS{h} g(\lambda)=\sum_{m=1}^{\infty}[\xi^m ]h(\xi)\partial_{\varphi_m}
    g(\lambda)= \sum_{m=1}^{\infty}[\xi^m ]h(\xi)\zav{g_\varphi
      \lambda^m+g_{\varphi^\prime} D_\lambda \lambda^m}=g_\varphi
    h(\lambda)+g_{\varphi^\prime} D_\lambda h(\lambda).
  \end{equation*}
  The last equality stems from the fact that $h$ is of the lowest order
  one. Finally,
  \begin{equation*}
    [\PBS{g},\PBS{h}]=\sum_{m=1}^{\infty}[\lambda^m ] \zav{\PBS{g}
      h(\lambda)-\PBS{h} g(\lambda)}\partial_{\varphi_m}.
  \end{equation*}
\end{proof}

\begin{proof}[Proof of Theorem~\ref{sec:hier-full-symm}]
  We have to prove the identity
  \begin{equation}
    [\PBS{f_{1+n}},\PBS{f_{1+m}}] = (n - m) \PBS{f_{1+n+m}}.
  \end{equation}
  Recall that
  \begin{equation*}
    \PBS{g}=\sum_{m=1}^{\infty}[\lambda^m ]g(\lambda)\partial_{\varphi_m},
  \end{equation*}
  where $g$ is a series in $\lambda$ of the lowest order one.

  For $n>1$, the quantity $ f_{1+n}=\lambda^{n+1} \varphi^{\prime}$ is of the lowest
  order one and depends on $\varphi^\prime$ only, hence is admissible for the
  commutation rule above. It is easy to show that for $n,m>1$ we have
  \begin{equation*}
    \szav{\lambda^{n+1} \varphi^{\prime},\lambda^{m+1} \varphi^{\prime}}
    =\lambda^{m+1}D_\lambda \lambda^{n+1}\varphi^\prime
    -\lambda^{n+1}D_\lambda
    \lambda^{m+1}\varphi^\prime=(n-m)\lambda^{m+n+1}\varphi^\prime. 
  \end{equation*}
  In other words
  \begin{equation*}
    [\PBS{f_{1+n}},\PBS{f_{1+m}}]=
    \PBS{(n-m)\lambda^{m+n+1}\varphi^{\prime}}=(n-m)\PBS{f_{1+m+n}}. 
  \end{equation*}
  
  For $n\leq 1$, dependence of $f_{1-n}$ on $\varphi^{(i)}$ is different,
  thus another approach is needed.  We are going to use the
  representation~\eqref{neargenfun} for $f_{1-n}(\lambda)$, that is,
  \begin{equation*}
    f_{1-n}(\lambda)=\frac{\varphi'(\lambda)}{\lambda^{n-1}} +
    [\xi^{n-2}]\frac{\varphi'(\xi)^2}{\varphi(\xi)-\varphi(\lambda)}. 
  \end{equation*}
  Note that $f_{1-n}(\lambda)$ does \textit{not} depend on $\xi$.

  Let us compute
  \begin{align*}
    &\PBS{f_{1-n}(\lambda)}f_{1-m}(\lambda)=
      \PBS{f_{1-n}(\lambda)}\frac{\varphi'(\lambda)}{\lambda^{m-1}}+[\xi^{m-2}]
      \PBS{f_{1-n}(\lambda)}
      \frac{\varphi'(\xi)^2}{\varphi(\xi)-\varphi(\lambda)}\\
    &\quad= \frac{1}{\lambda^{m-1}} f'_{1-n}(\lambda)+[\xi^{m-2}]f_{1-n}(\lambda)
      \frac{\varphi'(\xi)^2}{(\varphi(\xi)-\varphi(\lambda))^2} -[\xi^{m-2}]
      f_{1-n}(\xi)
      \frac{\varphi'(\xi)^2}{(\varphi(\xi)-\varphi(\lambda))^2}\\
    &\quad+[\xi^{m-2}]f'_{2-n}(\xi)\frac{2\varphi'(\xi)}{\varphi(\xi)-
      \varphi(\lambda)},\\
    \intertext{where the $'$ in $f'_{1-n}$ denotes the total
    derivative. Collecting terms that contain
    $\varphi'(\lambda)$, $\varphi''(\lambda)$ and the rest that depends only on
    $\varphi(\lambda)$, we can rewrite the last result as follows:}
    &\PBS{f_{1-n}(\lambda)}f_{1-m}(\lambda)=-(n-1)\frac{\varphi'(\lambda)}{\lambda^{m+n-1}}
      +\frac{\varphi''(\lambda)}{\lambda^{m+n-2}} \\
    &\quad
      +[\xi^{n-2}]\frac{\varphi'(\xi)^2}{(\varphi(\xi)-
      \varphi(\lambda))^2}\frac{\varphi'(\lambda)}{\lambda^{m-1}}+
      [\xi^{m-2}]\frac{\varphi'(\xi)^2}{(\varphi(\xi)-
      \varphi(\lambda))^2}\frac{\varphi'(\lambda)}{\lambda^{n-1}}
       +\sum_{k=0}^{\infty}c_k\varphi(\lambda)^k,
  \end{align*}
  where $c_k$ are some coefficients depending on $\varphi^{(i)}$.

  Notice that all terms except the first one and the last one are symmetric with respect 
  to swapping $n \leftrightarrow m$. Hence,
  \begin{align*}
    \szav{f_{1-n}(\lambda),f_{1-m}(\lambda)}&=
    \PBS{f_{1-n}(\lambda)}f_{1-m}(\lambda)-\PBS{f_{1-m}(\lambda)}f_{1-n}(\lambda)\\ 
    &=(m-n)\frac{\varphi'(\lambda)}{\lambda^{m+n-1}}+
    \sum_{k=0}^{\infty}c_k\varphi(\lambda)^k, 
  \end{align*}
  where, again, $c_k$ are some coefficients depending on $\varphi^{(i)}$,
  $i\geq1$.

  Note now that for any formal series $g,h$ that are of lowest order $1$, we
  have
  \begin{equation*}
    [\varphi(\lambda)^j]\szav{g(\lambda),h(\lambda)}=0,\qquad j\geq0, 
  \end{equation*}
  since the operator $\PBS{g}$ does not act on $\lambda$ and thus cannot
  produce non-positive powers of~$\lambda$.  In our case, both
  $f_{1-m}(\lambda)$ and $f_{1-n}(\lambda)$ are of lowest order $1$ and thus
  \begin{multline*}
    0=[\varphi(\lambda)^j]\szav{f_{1-n}(\lambda),f_{1-m}(\lambda)}=
    [\varphi(\lambda)^j]\zav{(m-n)\frac{\varphi'(\lambda)}{\lambda^{m+n-1}}+
      \sum_{k=0}^{\infty}c_k\varphi(\lambda)^k}\\  
    =(m-n)[\varphi(\lambda)^j]\frac{\varphi'(\lambda)}{\lambda^{m+n-1}}+c_j.
  \end{multline*}
  Hence,
  \begin{equation*}
    c_j=-(m-n)[\varphi(\lambda)^j]\frac{\varphi'(\lambda)}{\lambda^{m+n-1}}
  \end{equation*}
  and
  \begin{multline*}
    \szav{f_{1-n}(\lambda),f_{1-m}(\lambda)}\\
    =(m-n)\frac{\varphi'(\lambda)}{\lambda^{m+n-1}}-
    (m-n)\sum_{k=0}^{\infty}\varphi(\lambda)^k
    \hzav{\varphi(\xi)^k}\frac{\varphi'(\xi)}{\xi^{m+n-1}}
    \stackrel{\eqref{fngen}}{=}(m-n)f_{1-m-n}(\lambda).    
  \end{multline*} 
  Finally,
  \begin{equation*}
    \hzav{\PBS{f_{1-n}},\PBS{f_{1-m}}}=\PBS{\szav{f_{1-n},f_{1-m}}}=
    (-n-(-m))\PBS{f_{1-m-n}},
  \end{equation*}
  as claimed.
\end{proof}

\section{Uniqueness of symmetries}\label{sec:uniq-symm}

We prove here some uniqueness results for the nonlocal symmetries of the
Gibbons--Tsarev equation~\eqref{eq:21} in the covering~$\tau_*$ defined in
Section~\ref{sec:first-method}. For technical reasons, it is more convenient
for us to deal with System~\eqref{eq:29}, i.e.,
\begin{equation*}
  u_y+vu_x=\frac{1}{v-u},\qquad v_y+uv_x=\frac{1}{u-v}.
\end{equation*}
Due to~\eqref{eq:27}, the relation between shadows of~\eqref{eq:21} and those
of~\eqref{eq:29} is established by
\begin{equation}\label{eq:58}
  U=\frac{uD_x(Z)+D_y(Z)}{u-v},\quad V=\frac{vD_x(Z)+D_y(Z)}{v-u},
\end{equation}
where~$Z$ is a shadow for the Gibbons--Tsarev equation, while~$(U,V)$ is
a shadow for the system in~$u$ and~$v$. In particular, for the
symmetries~$\Si{-3},\dots,\Si{0}$ with the generating sections given
by~\eqref{eq:23} one has
\begin{align*}
  &\Zi{-3}&\mapsto&&& \Wi{-3}=\left(u_x,v_x\right),\\
  &\Zi{-2}&\mapsto&&& \Wi{-2}=\left(u_y,v_y\right),\\
  &\Zi{-1}&\mapsto&&& \Wi{-1}=\left(1-yu_x,1-yv_x\right),\\
  &\Zi{0}&\mapsto&&& \Wi{0}=\left(3xu_x+3yu_y-u,3xv_x+3yv_y-v\right), 
\end{align*}
while the symmetry~$\Si{-4}$ becomes invisible.

In what follows, we, for convenience, use the notation
\begin{equation}
  \label{eq:41}
  u_y=f(u,v)-vu_x,\qquad v_y=g(u,v)-uv_x,
\end{equation}
where~$f$ and~$g$ may be considered as functions in~$u$ and~$v$ such that the
partial derivatives~$f_u$, $f_v$, $g_u$, $g_v$ do not vanish. We choose the
functions
\begin{equation*}
  x,\ t,\ u_i=\frac{\partial^i u}{\partial x^i},\ v_i=\frac{\partial^i
    v}{\partial x^i},\qquad i\geq 0,
\end{equation*}
for the internal coordinates on~$\mathcal{E}$, and then the total derivatives
are
\begin{align*}
  D_x&=\frac{\partial}{\partial x}+\sum_{i\geq
    0}\left(u_{i+1}\frac{\partial}{\partial u_i} +
    v_{i+1}\frac{\partial}{\partial v_i}\right),\\
  D_y&=\frac{\partial}{\partial y} + \sum_{i\geq
    0}\left(D_x^i(f-vu_1)\frac{\partial}{\partial u_i} + 
    D_x^i(g-uv_1)\frac{\partial}{\partial v_i}\right).
\end{align*}
Then
\begin{equation*}
  \ell_{\mathcal{E}}=
  \begin{pmatrix}
    vD_x+D_y-f_u&u_1-f_v\\
    v_1-g_u&uD_x+D_y-g_v
  \end{pmatrix}
\end{equation*}
and the defining equations for symmetries of~\eqref{eq:41} are
\begin{equation}
  \label{eq:44}
  \begin{array}{rcl}
  D_y(U)&=&f_uU-vD_x(U)+(f_v-u_1)V,  \\
  D_y(V)&=&g_vV-uD_x(V)+(g_u-v_1)U.
  \end{array}
\end{equation}

\subsection{Uniqueness  of polynomial shadows}
\label{sec:uniq-polyn-shad}
Consider now the covering~$\tau_*\colon\mathcal{E}_*\to\mathcal{E}$ with the
nonlocal variables~$\psii{3},\dots,\psii{k},\dots$ defined in
Subsection~\ref{sec:first-method}. We say that a function~$F$
on~$\mathcal{E}_*$ is of \emph{order~$k$} if at least one of the partial
derivatives~$F_{u_k}$ or~$F_{u_k}$ does not vanish, while~$F_{u_i}=F_{u_i}=0$
for all~$i>k$.

Let us estimate the higher order terms of $\tau_*$-shadows.
The defining equations for $\tau_*$-shadows is obtained
from~\eqref{eq:44} by changing the total derivatives~$D_x$
and~$D_y$ to
\begin{equation*}
  \Dii{*}_x=D_x+\sum_{i\geq 3}\Xii{i}\frac{\partial}{\partial\psii{i}},\quad
  \Dii{*}_y=D_y+\sum_{i\geq 3}\Yii{i}\frac{\partial}{\partial\psii{i}},
\end{equation*}
i.e., they are of the form
\begin{align}\label{eq:46}
  \Dii{*}_y(U)&=f_uU-v\Dii{*}_x(U)+(f_v-u_1)V,
  \\\label{eq:47}
  \Dii{*}_y(V)&=g_vV-u\Dii{*}_x(V)+(g_u-v_1)U.
\end{align}
Note that the coefficients~$\Xii{i}$ and~$\Yii{i}$ are of order zero.

We shall need the following `asymptotics' below:
\begin{equation}
  \label{eq:48}
  \begin{array}{rcl}
  \bigl(\Dii{*}_x\bigr)^p(f-vu_1)&=&-vu_{p+1}+(f_u-pv_1)u_p+(f_v-u_1)v_p+O(p-1),
  \\
  \bigl(\Dii{*}_y\bigr)^p(g-uv_1)&=&-uv_{p+1}+(g_v-pu_1)v_p+(g_u-v_1)u_p+O(p-1)
  \end{array}
\end{equation}
for an arbitrary~$p>1$. Here and in what follows~$O(\alpha)$ denotes terms of
order~$\leq\alpha$.

\begin{proposition}\label{sec:uniq-polyn-shad-1}
  Equation~\eqref{eq:41} admits no $\tau_*$-shadow of order~$>1$.
\end{proposition}

\begin{proof}
  Let us assume that the components~$U$ and~$V$ of the shadow under
  consideration are of order~$k$ and, using~\eqref{eq:48}, differentiate
  Equation~\eqref{eq:46} with respect to~$v_{k+1}$. The result is
  $ -uU_{v_k}=-vU_{v_k}$.  In a similar way,
  applying~$\partial/\partial u_{k+1}$ to~\eqref{eq:47}, we get
  $ -vV_{u_k}=-uV_{u_k}$.  Consequently,
  \begin{equation}
    \label{eq:50}
    U=U(\dots,u_{k-1},v_{k-1},u_k),\qquad V=V(\dots,u_{k-1},v_{k-1},v_k),
  \end{equation}
  where `dots' stand for the variables of order~$\leq k-2$.

  Apply the partial derivatives~$\partial/\partial u_k$ and~$\partial/\partial
  v_k$ to Equations~\eqref{eq:46} and~\eqref{eq:47}:
  \begin{align*}
    \frac{\partial\eqref{eq:46}}{\partial u_k}:\ &
    \frac{\partial\Dii{*}_y}{\partial u_k}(U)+\Dii{*}_y(U_{u_k})=f_uU_{u_k}
    -v\biggl(\frac{\partial\Dii{*}_x}{\partial u_k}(U) +
      \Dii{*}_x(U_{u_k})\biggr) +(f_v-u_1)U_{u_k},\\
    \frac{\partial\eqref{eq:46}}{\partial v_k}:\ &
    \frac{\partial\Dii{*}_y}{\partial v_k}(U)+\Dii{*}_y(U_{v_k})=f_uU_{v_k}
    -v\biggl(\frac{\partial\Dii{*}_x}{\partial v_k}(U) +
      \Dii{*}_x(U_{v_k})\biggr) +(f_v-u_1)U_{v_k},\\
    \frac{\partial\eqref{eq:47}}{\partial u_k}:\ &
    \frac{\partial\Dii{*}_y}{\partial u_k}(V)+\Dii{*}_y(V_{u_k})=g_vV_{u_k}
    -u\biggl(\frac{\partial\Dii{*}_x}{\partial u_k}(V) +
      \Dii{*}_x(V_{u_k})\biggr) +(g_u-v_1)V_{u_k},\\
    \frac{\partial\eqref{eq:47}}{\partial v_k}:\ &
    \frac{\partial\Dii{*}_y}{\partial v_k}(V)+\Dii{*}_y(V_{v_k})=g_vV_{v_k}
    -u\biggl(\frac{\partial\Dii{*}_x}{\partial v_k}(V) +
      \Dii{*}_x(V_{v_k})\biggr) +(g_u-v_1)V_{v_k}
  \end{align*}
  (the partial derivatives above are applied to the coefficients of the
  corresponding operators). Using now~\eqref{eq:50}, we see that the above
  equalities amount to
  \begin{align*}
    \frac{\partial\Dii{*}_y}{\partial u_k}(U)+\Dii{*}_y(U_{u_k})&=
    f_uU_{u_k}-v\biggl(\frac{\partial\Dii{*}_x}{\partial
        u_k}(U)+\Dii{*}_x(U_{u_k})\biggr),\\
    \frac{\partial\Dii{*}_y}{\partial v_k}(V)+\Dii{*}_y(V_{v_k})&=
    g_vV_{v_k}-u\biggl(\frac{\partial\Dii{*}_x}{\partial
        v_k}(V)+\Dii{*}_x(V_{v_k})\biggr) \intertext{and}
    \frac{\partial\Dii{*}_y}{\partial v_k}(U)&=
    -v\frac{\partial\Dii{*}_x}{\partial v_k}(U)+(f_v-u_1)V_{v_k},\\
    \frac{\partial\Dii{*}_y}{\partial u_k}(V)&=
    -u\frac{\partial\Dii{*}_x}{\partial u_k}(V)+(g_u-v_1)U_{u_k}.
  \end{align*}
  Now, by~\eqref{eq:48}, we have
  \begin{align*}
    &\frac{\partial\Dii{*}_x}{\partial u_k}=\frac{\partial}{\partial u_{k-1}},
    &&\frac{\partial\Dii{*}_y}{\partial
      u_k}=(f_u-kv_1)\frac{\partial}{\partial
      u_k}+(g_u-v_1)\frac{\partial}{\partial v_k}-v\frac{\partial}{\partial
      u_{k-1}},\\
    &\frac{\partial\Dii{*}_x}{\partial v_k}=\frac{\partial}{\partial v_{k-1}},
    &&\frac{\partial\Dii{*}_y}{\partial v_k}=(f_u-u_1)\frac{\partial}{\partial
      u_k}+(g_v-ku_1)\frac{\partial}{\partial v_k}-u\frac{\partial}{\partial
      v_{k-1}}.
  \end{align*}
  and, using~\eqref{eq:50} again, we arrive to
  \begin{align}\nonumber
    -kv_1U_{u_k}+\Dii{*}_y(U_{u_k})&=-v\Dii{*}_x(U_{u_k}),\\
    \nonumber
    -ku_1V_{v_k}+\Dii{*}_y(V_{v_k})&=-u\Dii{*}_x(V_{v_k})
    \intertext{and}
    (f_v-u_1)U_{u_k}+(v-u)U_{v_{k-1}}&=(f_v-u_1)V_{v_k},\label{eq:51}\\
    (g_u-v_1)V_{v_k}+(u-v)V_{u_{k-1}}&=(g_u-v_1)U_{u_k}.\label{eq:52}
  \end{align}
  Then, differentiating Equation~\eqref{eq:51} with respect to~$v_k$ and
  Equation~\eqref{eq:52} with respect to~$u_k$, we obtain
  \begin{equation*}
    (f_v-u_1)V_{v_kv_k}=(g_u-v_1)U_{u_ku_k}=0,
  \end{equation*}
  i.e.,
  \begin{equation}
    \label{eq:53}
    U=au_k+b,\qquad V=bv_k+d,
  \end{equation}
  where the functions~$a$, $b$, $c$, and~$d$ are of order~$k-1$. Let us
  substitute the obtained expressions~\eqref{eq:53} to the defining
  system~\eqref{eq:46}--\eqref{eq:47}:
  \begin{align}
    \label{eq:54}
    \Dii{*}_y(a)u_k+a\Dii{*}_y(u_k)+\Dii{*}_y(b)&=f_u(au_k+b)
    -v\left(\Dii{*}_x(a)u_k+au_{k+1}+\Dii{*}(b)\right)\\\nonumber
    &\quad +(f_v-u_1)(cv_k+d) \\\label{eq:55}
    \Dii{*}_y(c)v_k+c\Dii{*}_y(v_k)+\Dii{*}_y(d)&=g_v(cv_k+d)
    -u\left(\Dii{*}_x(c)v_k+cv_{k+1}+\Dii{*}(d)\right) \\\nonumber
    &\quad +(g_u-u_1)(au_k+b).
  \end{align}
  Using the estimates~\eqref{eq:48} and comparing the terms
  containing~$u_{k+1}$, $v_{k+1}$ and~$u_k$, $v_k$, we see that the terms
  with~$u_{k+1}$, $v_{k+1}$ and~$u_k^2$, $v_k^2$ are cancelling, while
  \begin{align}\nonumber
    &\text{in Eq.~\eqref{eq:54} at $u_kv_k$:}
    &&{-}ua_{v_{k-1}}=-va_{v_{k-1}},\\\nonumber 
    &\text{in Eq.~\eqref{eq:55} at $u_kv_k$:}
    &&{-}vc_{u_{k-1}}=-uc_{u_{k-1}},\\\label{eq:56}
    &\text{in Eq.~\eqref{eq:54} at $u_k$:}
    &&{-}kv_1a=0,\\\label{eq:57}
    &\text{in Eq.~\eqref{eq:55} at $v_k$:}
    &&{-}ku_1c=0,\\\nonumber
    &\text{in Eq.~\eqref{eq:54} at $v_k$:}
    &&(f_v-u_1)(a-c)=(u-v)b_{v_{k-1}},\\\nonumber
    &\text{in Eq.~\eqref{eq:55} at $u_k$:}
    &&(g_u-v_1)(c-a)=(v-u)d_{u_{k-1}}.
  \end{align}
  In particular, from Equations~\eqref{eq:56} and~\eqref{eq:57} we see that
  the coefficients~$a$ and~$c$ vanish and thus, by virtue of~\eqref{eq:53},
  the functions~$U$ and~$V$ are of order~$k-1$. We repeat the procedure until
  the order of the shadows at hand becomes equal to~$1$.
\end{proof}

Using Proposition~\ref{sec:uniq-polyn-shad-1}, we shall now prove that the
symmetries~$\Si{i}=\Ev_{\Zi{i}}$, $i=-4,-3,\dots$, exhaust all the polynomial
symmetries in the covering~$\tau_*$.

\begin{theorem}
  Any $\tau_*$-nonlocal symmetry of the Gibbons--Tsarev equation of
  weight~$k$\textup{,} polynomial in all variables\textup{,} coincides
  with~$\Si{k}$ up to a constant factor\textup{,} $k\geq-4$.
\end{theorem}

\begin{proof}
  Using Equations~\eqref{eq:58}, we pass from symmetries of the Gibbons--Tsarev
  equation~\eqref{eq:21} to those of System~\eqref{eq:29}. The proof is
  accomplished by induction on the weight.

  For small weights ($\wt{\mathcal{S}}=-4,\dots,0$) this fact can be checked
  by direct computations due to Proposition~\ref{sec:uniq-polyn-shad-1}.

  Let us fix a~$k>0$ and assume that for all weights less than~$k$ the
  statement is true. To proceed with the proof, we need a number of auxiliary
  facts. The first two of them can be observed from the results of
  Section~\ref{sec:algebra-nonl-symm}.

  \begin{fact}
    For a symmetry~$\Si{i}=\Ev_{\Zi{i}}$, one has the following `asymptotics'
    in~$\psi$s:
    \begin{equation*}
      \Zi{1}=-3\psii{4}+\text{ local terms}
    \end{equation*}
    and, for~$i>1$,
    \begin{equation*}
      \Zi{i}=-(i+2)\psii{i+3}+\frac{2i+3}{2}z_x\psii{i+2}+\Upsilon(i+1),\qquad
      i\geq 0,
    \end{equation*}
    where~$\Upsilon(\alpha)$ denotes the terms independent of~$\psii{\beta}$
    for~$\beta>\alpha$. This means, by~\eqref{eq:58}, that the corresponding
    generating section for the system is of the form
    \begin{equation*}
      \Wi{1}=\frac{5}{2}\psii{3}\Wi{-3}+\text{ local terms}
    \end{equation*}
    and, for~$i>1$,
    \begin{align*}
      \Ui{i}&=\frac{2i+3}{2}\psii{i+2}\Ui{-3}+(i+1)\psii{i+1}\Ui{-2}
      +\Upsilon(i),\\ 
      \Vi{i}&=\frac{2i+3}{2}\psii{i+2}\Vi{-3}+(i+1)\psii{i+1}\Vi{-2}
      +\Upsilon(i),
    \end{align*}
    or
    \begin{equation}
      \label{eq:59}
      \Wi{i}=\frac{2i+3}{2}\psii{i+2}\Wi{-3}+(i+1)\psii{i+1}\Wi{-2}
      +\Upsilon(i),
    \end{equation}
    where
    \begin{equation*}
      \Wi{-3}=(u_1,v_1),\qquad
      \Wi{-2}=\left(\frac{1}{v-u}-vu_1,\frac{1}{u-v}-uv_1\right),
    \end{equation*}
    are the generating sections of the infinitesimal $x$- and
    $y$-translations, respectively.
  \end{fact}

  \begin{fact}
    The shadow~$\Wi{-1}=(1-yu_1,1-yv_1)$ of the generalised Galilean boost
    extends to~$\mathcal{E}_*$ as follows
    \begin{multline*}
      \Si{-1}=\sum_{l\geq0}\left(D_x^l(1-yu_1)\frac{\partial}{\partial u_l} +
        D_x^l(1-yv_1)\frac{\partial}{\partial v_l}\right) +
      \left(2x-y\Xii{3}\right)\frac{\partial}{\partial \psii{3}}\\
      +\sum_{j>3}\left((j-1)\psii{j-1}
        -y\Xii{j}\right)\frac{\partial}{\partial\psii{j}}. 
    \end{multline*}
    Since the last expression can be rewritten in the form
    \begin{equation*}
      \Si{-1}=\frac{\partial}{\partial u} + \frac{\partial}{\partial v} +
      y\frac{\partial}{\partial x} - y\Dii{*}_x
      +2x\frac{\partial}{\partial\psii{3}} +
      \sum_{j>3}(j-1)\psii{j-1}\frac{\partial}{\partial\psii{j}},
    \end{equation*}
    while
    \begin{equation*}
      \ell_{\Wi{-1}}=
      \begin{pmatrix}
        - y\Dii{*}_x&0\\
        0&- y\Dii{*}_x
      \end{pmatrix},
    \end{equation*}
    we obtain  
    \begin{equation}
      \label{eq:60}
      \{\Wi{-1},W\}=
      \begin{pmatrix} \displaystyle
        \dfrac{\partial U}{\partial u} + \dfrac{\partial U}{\partial v} +
        y\dfrac{\partial U}{\partial x} +2x\dfrac{\partial U}{\partial\psii{3}}
        +\sum_{j>3}(j-1)\psii{j-1}\dfrac{\partial U}{\partial\psii{j}}
      \\[17pt] \displaystyle
        \dfrac{\partial V}{\partial u} + \dfrac{\partial V}{\partial v}+
        y\dfrac{\partial V}{\partial x} +2x\dfrac{\partial V}{\partial\psii{3}}
        +\sum_{j>3}(j-1)\psii{j-1}\dfrac{\partial V}{\partial\psii{j}}
      \end{pmatrix}
    \end{equation}
    for any~$W=(U,V)$, and, in particular,
    \begin{equation*}
      \{\Wi{-1},\Wi{i}\}=(i+1)\Wi{i-1}+\Upsilon(i-2),\qquad i\geq2.
    \end{equation*}
  \end{fact}

  \begin{fact}
    A straightforward, but important consequence of~\eqref{eq:60} is that the
    adjoint action~$W\mapsto\{\Wi{-1},W\}$ is a derivation, i.e.,
    \begin{equation*}
      \{\Wi{-1},hW\}=h\{\Wi{-1},W\}+X_{\Wi{-1}}(h)W,\qquad
      h\in\mathcal{F}(\mathcal{E}_*),
    \end{equation*}
    where
    \begin{equation*}
      X_{\Wi{-1}}=\frac{\partial}{\partial u} + \frac{\partial}{\partial v} +
        y\frac{\partial}{\partial x} +2x\frac{\partial }{\partial\psii{3}}
        +\sum_{j>3}(j-1)\psii{j-1}\frac{\partial}{\partial\psii{j}}
    \end{equation*}
    is a vector field on~$\mathcal{E}$.
  \end{fact}
  
  \begin{fact}
    Let~$W=(U,V)$ be a solution of System~\eqref{eq:46}-\eqref{eq:47} of
    weight~$k$. Let also~$l$ be the minimal integer such that~$\partial
    U/\partial \psii{j}$ and~$\partial V/\partial \psii{j}$ vanish for all~$j>
    l$. Recall (see Subsection~\ref{sec:first-method}) that the `nonlocal
    tails' of the total derivatives~$\Dii{*}_x$ and~$\Dii{*}_y$
    on~$\mathcal{E}_*$ are of the form
    \begin{equation*}
      X=\sum_{j\geq3}\Xii{j}\frac{\partial}{\partial\psii{j}},\qquad
      Y=\sum_{j\geq3}\Yii{j}\frac{\partial}{\partial\psii{j}},
    \end{equation*}
    where
    \begin{equation*}
      \Xii{j}=-(j-3)\psii{j-3}+\Upsilon(j-4),\qquad
      \Yii{j}=-(j-2)\psii{j-2}+\Upsilon(j-4).
    \end{equation*}
    This implies that if~$W=(U,V)$ and~$l$ is chosen as above, then
    \begin{equation*}
      \frac{\partial W}{\partial\psii{l}}=\left(\frac{\partial
          U}{\partial\psii{l}},\frac{\partial
          V}{\partial\psii{l}}\right),\qquad
      \frac{\partial W}{\partial\psii{l-1}}=\left(\frac{\partial
          U}{\partial\psii{l-1}},\frac{\partial
          V}{\partial\psii{l-1}}\right)
    \end{equation*}
    are shadows as well (of weights~$k-l-1$ and~$k-l$, respectively).
  \end{fact}

  Let us now return to the main course of the proof. Since~$k-l<k$, then due
  to the induction hypothesis we have
  \begin{equation}
    \label{eq:61}
    \frac{\partial W}{\partial\psii{l}}=\alpha\Wi{k-l-1}=
    \alpha\left(\frac{2k-2l+1}{2}\psii{k-l+1}\Wi{-3} +
      (k-l)\psii{k-l}\Wi{-2}\right) + \Upsilon(k+l-1),
  \end{equation}
  where~$\alpha\in\mathbb{R}$ is a nonvanishing constant. Note also that due
  to the definition of~$l$ one has~$l\geq k-l+1$, or
  \begin{equation}
    \label{eq:62}
    2l\geq k+1.
  \end{equation}
  We now consider two cases: Inequality~\eqref{eq:62} is either strict or 
  an equality.

  \textbf{The case~$2l>k+1$.} In this case, Equation~\eqref{eq:61} implies
  \begin{equation}
    \label{eq:63}
    W=\alpha\Wi{k-l-1}\psii{l}+\Upsilon(l-1).
  \end{equation}
  Let us apply the operator~$\{\Wi{-1},\cdot\}$ to both sides of~\eqref{eq:63}:
  \begin{multline*}
    \{\Wi{-1},W\}=\alpha\{\Wi{-1},\Wi{k-l-1}\}\psii{l}+
    \alpha\Wi{k-l-1}X_{\Wi{-1}}(\psii{l})+\Upsilon(l-2)\\
    =\alpha\{\Wi{-1},\Wi{k-l-1}\}\psii{l}+
    \alpha(l-1)\Wi{k-l-1}\psii{l-1}+\Upsilon(l-2).
  \end{multline*}
  But~$\wt{\{\Wi{-1},W\}}=k-1$ and, by the induction hypothesis, we must
  have~$\{\Wi{-1},W\}=\beta\Wi{k-1}+\Upsilon(k-2)$. Consequently,
  using~\eqref{eq:59} and~\eqref{eq:60}, we obtain
  \begin{equation*}
    \beta\left(\frac{2k+1}{2}\Wi{-3}\psii{k+1}+\Upsilon(k)\right)=
    \alpha(k-l)\Wi{k-l-2}\psii{l} +\alpha(l-2)\Wi{k-l-1}\psii{l-1}
    +\Upsilon(l-2).
  \end{equation*}
 The last equality can hold only when
 \begin{equation*}
   l=k+1,\qquad l=-\frac{2k+1}{2}\beta.
 \end{equation*}
 Consider the shadow~$\tilde{W}=W-\beta\Wi{k}$. There are two possibilities:
 (a)~$\tilde{W}=0$ and then the proof is finished; (b)~$\tilde{W}\neq0$ and
 then there should exist the minimal integer~$\tilde{l}< l$ such
 that~$\partial\tilde{W}/\partial\psii{\tilde{l}}\neq0$. The only possibility
 is~$\tilde{l}=(k+1)/2$ and thus we pass to the second case.

 \textbf{The case~$2l=k+1$.} Now Equation~\eqref{eq:61} reads
 \begin{equation*}
   \frac{\partial W}{\partial\psii{l}} = \alpha\frac{2l-1}{2}\psii{l}\Wi{-3} +
   \Upsilon(l-1),
 \end{equation*}
 or
 \begin{equation*}
   W=\alpha\frac{2l-1}{4}\left(\psii{l}\right)^2\Wi{-3} + \text{ terms
     linear in } \psii{l} + \Upsilon(l-1).
 \end{equation*}
 Let us apply the operator~$\{\Wi{-1},\cdot\}$ to the last equation. Then in
 the left-hand side we obtain a shadow of weight~$k-1=2l-2$ which, by the
 induction hypothesis and Equation~\eqref{eq:59}, must be proportional
 to~$\psii{2l}\Wi{-3}+\Upsilon(2l-1)$. But in the right-hand side such a term
 cannot appear. This contradiction finishes the proof.
\end{proof}

\subsection{Uniqueness of invisible symmetries}
\label{sec:uniq-invis-symm}

Consider a symmetry
\begin{equation*}
  \mathcal{S}=
  \sum_{i\geq0}\left(\left(\Dii{*}_x\right)^i(U)\frac{\partial}{\partial 
      u_i} + \left(\Dii{*}_x\right)^i(V)\frac{\partial}{\partial v_i}\right) +
  \sum_{\alpha\geq 3}\Psii{\alpha }\frac{\partial}{\partial\psii{\alpha }}
\end{equation*}
of~$\mathcal{E}_*$. Let us say that~$S$ is \emph{invisible of depth~$k$} if
$ U=V=\Psii{3}=\dots=\Psii{k-1}=0$, i.e.,
\begin{equation*}
  \mathcal{S}=\sum_{\alpha\geq k}\Psii{\alpha
  }\frac{\partial}{\partial\psii{\alpha }}. 
\end{equation*}
The defining equations for such symmetries are
$ [\mathcal{S},\Dii{*}_x]=[\mathcal{S},\Dii{*}_y]=0$, or
\begin{align}\label{eq:64}
  \Dii{*}_x(\Psii{\alpha})=
  \sum_{\beta=k}^{\alpha-3}
  \Psii{\beta}\frac{\partial\Xii{\alpha}}{\partial\psii{\beta}}, 
  \qquad
  \Dii{*}_y(\Psii{\alpha})=
  \sum_{\beta=k}^{\alpha-2}
  \Psii{\beta}\frac{\partial\Yii{\alpha}}{\partial\psii{\beta}} 
\end{align}
for all~$\alpha\geq k$, where, as before,
\begin{equation*}
  \Dii{*}_x=D_x+ \sum_{\alpha\geq
    3}\Xii{\alpha}\frac{\partial}{\partial\psii{\alpha}},\quad 
  \Dii{*}_y=D_y+\sum_{\alpha\geq
    3}\Yii{\alpha}\frac{\partial}{\partial\psii{\alpha}} 
\end{equation*}
and~$\Xii{\alpha}$, $\Yii{\alpha}$ are the right-hand sides of~\eqref{eq:69}
and~\eqref{eq:70}, respectively.

\begin{theorem}
  Any nontrivial invisible symmetry of depth~$k$ is of the form
  \begin{equation*}
    \frac{\partial}{\partial\psii{k}} +
    \gamma\frac{\partial}{\partial\psii{k+1}} + \sum_{\alpha\geq
      k+2}\Psii{\alpha}\frac{\partial}{\partial\psii{\alpha}}, 
  \end{equation*}
  where~$\gamma=\const$.
\end{theorem}

\begin{proof}
  Indeed, the right-hand sides of Equations~\eqref{eq:64} vanish
  for~$\alpha=k$ and~$\alpha=k+1$, i.e.,
  \begin{equation*}
    \Dii{*}_x(\Psii{k})=\Dii{*}_x(\Psii{k+1})=0,\quad
    \Dii{*}_y(\Psii{k})=\Dii{*}_y(\Psii{k+1})=0.
  \end{equation*}
  But, by Proposition~\ref{sec:first-method-2}, the equation~$\mathcal{E}_*$
  is differentially connected and thus~$\Psii{k}$ and~$\Psii{k+1}$ are
  constants.
\end{proof}

\begin{remark}
  Actually, one can say more about the structure of the
  coefficients~$\Psii{\alpha}$ (see Equation~\eqref{S1-n}), but for our cause
  the above said is sufficient.
\end{remark}



\section*{Acknowledgments}

This research was undertaken within the framework of the 
OPVK programme, project CZ.1.07/\allowbreak2.300/\allowbreak20.0002.
The work of I.S. Krasil{\cprime}shchik was partially supported by the RFBR
Grant 18-29-10013 and the Simons-IUM fellowship.
The work of P. Blaschke and M. Marvan was partially supported by GA\v{C}R under
project P201/12/G028.

The symbolic computations were performed with the aid of the \texttt{Jets} 
software~\cite{Jets}. 
The authors are grateful to Jenya Ferapontov who draw our attention to the
problem. We also want to express our gratitude to Maxim Pavlov, Vladimir
Sokolov, and Sergey Tsarev for fruitful discussions.

\end{document}